\documentclass[aps,pra,showpacs,amssymb,superscriptaddress,twocolumn,nofootinbib]{revtex4-2}
\pdfoutput=1
\usepackage{fullpage}
\usepackage{hyperref}
\usepackage[margin=2cm]{geometry}

\usepackage{microtype}
\usepackage{enumitem}

\usepackage{amsmath}
\usepackage{amsthm}
\usepackage{amssymb}
\usepackage{dsfont}
\usepackage{url}
\usepackage{hyperref}
\hypersetup{
    colorlinks=true,
    linkcolor=blue,
    citecolor=blue,
    filecolor=blue,
    urlcolor=blue
}
\usepackage{graphicx}
\usepackage{caption}
\usepackage{subcaption}
\usepackage{wrapfig}
\usepackage[normalem]{ulem}

\usepackage{physics}
\usepackage{array}

\newtheorem{theorem}{Theorem}

\newtheorem{proposition}{Proposition}
\newtheorem{observation}{Observation}

\usepackage{tikz}
\usetikzlibrary{arrows.meta, positioning}

\usepackage{tcolorbox}

\begin{document}

\title{Symmetric Localizable Multipartite Quantum Measurements from Pauli Orbits}
\author{Jef Pauwels}
\email{jef.pauwels@unige.ch}
\affiliation{Department of Applied Physics, University of Geneva, Switzerland}
\affiliation{Constructor University, 28759 Bremen, Germany}
\author{Cyril Branciard}
\affiliation{Univ.\ Grenoble Alpes, CNRS, Grenoble INP\footnote{Institute of Engineering Univ.\ Grenoble Alpes}, Institut N\'eel, 38000 Grenoble, France}
\author{Alejandro Pozas-Kerstjens}
\affiliation{Department of Applied Physics, University of Geneva, Switzerland}
\author{Nicolas Gisin}
\affiliation{Department of Applied Physics, University of Geneva, Switzerland}
\affiliation{Constructor University, 28759 Bremen, Germany}

\begin{abstract}
While the structure of entangled quantum states is relatively well understood, the characterization of entangled measurements, especially in multipartite and high-dimensional settings, remains far less developed. In this work, we introduce a general approach to construct highly symmetric, locally encodable orthonormal measurement bases, as orbits of a single fiducial state under tensor-product actions of Pauli subgroups. This framework recovers the Elegant Joint Measurement—a two-qubit measurement whose local marginals form a regular tetrahedron on the Bloch sphere—as a special case, and we extend the construction to both more systems and higher dimensions. We analyze the entanglement cost required to implement these measurements locally via the Clifford hierarchy and use this criterion to classify them. We show how the symmetry of our constructions allows us to characterize their localizability, which is generally a challenging problem, and to identify certain classes of measurement bases that are efficiently localizable.  Our approach offers a systematic toolkit for designing entangled measurements with rich symmetry and implementability properties.
\end{abstract}

\maketitle

\section{Introduction}

Entanglement lies at the heart of quantum information science~\cite{Horodecki2009}, and the study of entangled states has driven decades of progress across quantum computing~\cite{Jozsa2003}, communication~\cite{Vazirani2014,Tavakoli2021a}, and foundations~\cite{Bell1988}. In contrast, far less is known about the structure, classification, and implementation of entangled measurements and nonlocal observables~\cite{Cavalcanti2023}. While much attention has focused on the preparation and manipulation of entangled states, the role of entangled measurements in quantum information—particularly multipartite measurement bases—indeed remains little explored~\cite{DelSanto2024,Pauwels2025}.

Filling this gap is not merely of theoretical interest. In many quantum information tasks, such as quantum teleportation~\cite{Bennet1993}, network nonlocality~\cite{Tavakoli2022}, distributed quantum computing~\cite{Caleffi2024}, quantum communication~\cite{Piveteau2022,zhang2025} and quantum network routing~\cite{Briegel1998}, the ability to perform joint entangled measurements is essential. Understanding how to construct such measurements with desirable properties—symmetry, constraints on their entanglement, and implementability, both in terms of local encodability and localizability (low-complexity realization with shared entanglement)—could unlock new insights.

In this work, we take a concrete step in this direction by revisiting a paradigmatic example: the Elegant Joint Measurement (EJM), introduced in Ref.~\cite{Gisin2019}. The EJM is a partially entangled two-qubit orthonormal basis whose single-qubit marginals form a regular tetrahedron on the Bloch sphere. It is isoentangled~\cite{DelSanto2024} and constitutes one of the few known entangled measurements that can be locally implemented at low entanglement cost~\cite{Pauwels2025}.
Starting from this example, we ask which of its properties can be generalized. We identify Pauli covariance, geometric symmetry, local encodability~\cite{Tanaka007}, and efficient localizability~\cite{Pauwels2025} as core properties and ask whether these can be extended to larger multipartite systems.

We describe a method for constructing multipartite orthonormal measurement bases in any uniform local dimension, using group orbits of a fiducial state under tensor-product actions of the generalized Pauli group. All such bases are locally encodable and exhibit local tetrahedral symmetry by construction. Through this approach, in particular we:
\begin{itemize}
\item Recover the two-qubit EJM as a special case,
\item Identify a family of party-permutation invariant (PPI) qubit bases with tetrahedral symmetry for all odd numbers of qubits,
\item Find three-qubit bases with regular tetrahedral Bloch vectors, extending the EJM structure,
\item Construct real-valued bases with rectangular symmetry for an arbitrary number of qubits,
\item Recover the higher-dimensional generalizations of the EJM based on Weyl–Heisenberg symmetry~\cite{Czartowski2021},
\item Provide a simple characterization of which bases obtained through our construction are efficiently localizable—an otherwise difficult problem that becomes tractable here due to their symmetry.
\end{itemize}

Along the way, we connect the geometric structure of the single-qubit reductions to the algebraic properties of the fiducial state, and identify which group-theoretic conditions ensure that the basis formed by Pauli orbits is complete and orthonormal. Our results point toward a deeper understanding of the intersection between geometry, entanglement, and measurement theory.

\begin{figure*}[ht!]

\begin{tcolorbox}[width=0.8\textwidth,colback=gray!5!white, colframe=black!75, title=Paper Structure Overview, fonttitle=\bfseries, boxrule=0.5pt]

\begin{itemize}
  \item[\textbf{[\ref{subsec:standard_EJM}]}] \textbf{The Elegant Joint Measurement (EJM)} \\
  \hspace*{1em}-- Group orbit under $G_{\mathrm{tetra}}^{(2)} = \langle Z \otimes Z,\, X \otimes X \rangle$ \\
  \hspace*{1em}-- Tetrahedral Bloch-vector symmetry

  \item[\textbf{$\Downarrow$}] \emph{Abstract via group-theoretic formulation}

  \item[\textbf{[\ref{subsec:tetrahedral_bases}]}] \textbf{2-Qubit Tetrahedral Bases} \\
  \hspace*{1em}-- Entire family from $G_{\mathrm{tetra}}^{(2)}$-orbits of fiducial states \\
  \hspace*{1em}-- Basis completeness via Weyl twirling

  \item[\textbf{$\Downarrow$}] \emph{Generalization in two directions}

  \item[\textbf{[\ref{sec:multiqubit}]}] \textbf{Multi-Qubit Tetrahedral Bases} \\
  \hspace*{1em}-- Define $G_{\mathrm{tetra}}^{(n)} = \langle Z_i Z_{i+1},\, X^{\otimes n} \rangle \cong \mathbb{Z}_2^n$ \\
  \hspace*{1em}-- fiducial states = equal superposition of joint eigenstates
  \begin{itemize}
    \item[\textbf{[\ref{subsec:PPI}]}] \textbf{Permutationally Symmetric (PPI) Subfamily} \\
    \hspace*{2em}-- Dicke decomposition + bit-flip symmetry constraint
    \item[\textbf{[\ref{subsec:regular_geoms}]}] \textbf{Regular Geometric Subfamily} \\
    \hspace*{2em}-- Regular tetrahedral directions (3-qubit) \\
    \hspace*{2em}-- Rectangular (planar) Bloch-vector arrangements
  \end{itemize}

  \item[\textbf{[\ref{sec:higher-dim}]}] \textbf{Higher-Dimensional (Qudit) Generalizations} \\
  \hspace*{1em}-- $G_d^{(n)} = \langle Z_d^{(i)} Z_d^{(i+1)^{-1}},\, \dots,\, X_d^{\otimes n} \rangle \cong \mathbb{Z}_d^n$ \\
  \hspace*{1em}-- SUM-conjugated eigenbases yield generalized Bell states \\
  \hspace*{1em}-- fiducial states yield isoentangled, locally WH covariant configurations \\
  \hspace*{1em}-- Higher dimensional analogue of EJM~\cite{Czartowski2021} as a special case
 \item[\textbf{[\ref{sec:Localizability}]}]  \textbf{Localizability and Clifford Hierarchy} \\
    \hspace*{2em}-- Entanglement-assisted local implementation cost \\
    \hspace*{2em}-- EJM lies at level 2; explore other multi-qubit bases
\end{itemize}

\end{tcolorbox}

\centering\caption{High-level overview of the paper structure.}
\label{overview}
\end{figure*}

We conclude by discussing possible generalizations beyond the Pauli group, and by raising the question of whether other finite symmetry groups can yield locally covariant and efficiently localizable measurements. This opens new avenues for exploring the rich but largely uncharted terrain of structured entangled measurements. A high-level overview of the paper is given in Figure \ref{overview}.

\section{Group Covariant Measurements}

A powerful and widely studied~\cite{Holevo1982,Chiribella2004,Renes2004,Waldron2018} approach to constructing quantum measurements—or orthonormal bases for projective measurements—with desirable symmetry properties is to exploit the structure of group representations. Specifically, given a group \( G \) and a unitary representation \( U_G: G \rightarrow \mathcal{U}(\mathcal{H}),\, g\mapsto U_g \) acting on a Hilbert space \( \mathcal{H} \), one can construct an ensemble of quantum states from a single \emph{fiducial state}—that is, a reference state from which the rest of the measurement elements are generated—via its group orbit:
\begin{equation} \label{eq:basisdef}
    \ket{\psi_g} = U_g \ket{\psi}, \qquad \forall\,g \in G.
\end{equation}

The goal is to choose a fiducial state \( \ket{\psi} \), as well as an appropriate group \( G \) and representation \( U_G \), such that the resulting set \( \{ \ket{\psi_g} \}_{g \in G} \) defines a valid quantum measurement. In the case of projective measurements, this requires the resulting set to form an orthonormal basis. More generally, for a POVM, one demands the \emph{completeness condition}:
\begin{equation} \label{eq:completeness}
    \sum_{g \in G} U_g \ketbra{\psi}{\psi} U_g^\dagger = c \, \mathds{1},
\end{equation}
for some constant \( c > 0 \), which ensures that the rescaled set of operators forms a valid POVM. This condition asserts that the twirling of the fiducial state yields (a scalar multiple of) the maximally mixed state, meaning that the group action fully ``scrambles'' the fiducial state across the Hilbert space.

Measurements constructed in this manner are referred to as group covariant measurements~\cite{Chiribella2004}, and are particularly useful when the group \( G \) reflects a physical symmetry of the system. These constructions apply both to finite groups and to continuous Lie groups, in which case the summation in Eq.~\eqref{eq:completeness} is replaced by an integral with respect to the Haar measure, and the POVM may involve continuously many outcomes. Group-covariant measurements have been studied extensively in both settings: for compact Lie groups (e.g., \( SU(2), U(1) \))~\cite{Chiribella2004}, and for finite groups, where they connect to the theory of group frames~\cite{Renes2004,Waldron2018}.

A notable example is the construction of symmetric informationally complete POVMs (SIC-POVMs)~\cite{Lemmens1973}, which consist of \( d^2 \) unit vectors in a \( d \)-dimensional Hilbert space whose pairwise overlaps all have equal modulus. Zauner’s conjecture~\cite{Zauner1999, Renes2004} posits that such SIC-POVMs exist in every finite dimension and can be constructed as group orbits of a single fiducial vector under the action of the Weyl–Heisenberg group.

A central theoretical tool enabling these constructions is Schur’s lemma~\cite{Schur1905}. If \( U_G \) is an \emph{irreducible} representation—that is, one which admits no nontrivial invariant subspaces—then for any operator \( \mathcal{O} \) acting on \( \mathcal{H} \),
\begin{equation}
    \sum_{g \in G} U_g \mathcal{O} U_g^\dagger = \alpha \, \mathds{1},
\end{equation}
for some scalar \( \alpha \). Applied to the rank-one projector \( \mathcal{O} = \ketbra{\psi}{\psi} \), this implies that the completeness condition is automatically satisfied (with $c=\alpha$) for any choice of fiducial state whenever the representation is irreducible.

In the more general case where the representation is reducible (but multiplicity-free), Schur’s lemma implies 
\begin{equation} \label{eq:shur}
   \sum_{g \in G} U_g \mathcal{O} U_g^\dagger = \bigoplus_{\lambda} \alpha_\lambda \, \mathds{1}_\lambda \,,
\end{equation}
where the direct sum runs over the irreducible subspaces~\footnote{An irreducible subspace is one that carries an irreducible representation of \( G \), i.e., it does not contain any proper invariant subspaces.} $\mathcal{H}_\lambda$ of \( \mathcal{H} \),  \( \mathds{1}_\lambda \) denotes the identity operator on  $\mathcal{H}_\lambda$ and $\alpha_\lambda$ is a scalar attached to $\mathcal{H}_\lambda$. For the rank-one projector \( \mathcal{O} = \ketbra{\psi}{\psi} \) onto a fiducial state $\ket{\psi}$, the completeness condition imposes nontrivial constraints: it must have equal values $\alpha_\lambda$ in each irreducible component in order for the right-hand side of Eq.~\eqref{eq:shur} to be proportional to the identity on the full space. 
As it follows from Eq.~\eqref{eq:shur} that
\begin{align}
\alpha_\lambda 
&= \Tr\left[{\textstyle \sum_g} U_g^\dagger \,\Pi_\lambda\, U_g \,\ketbra{\psi}\right]/d_\lambda \nonumber \\
&= (|G|/d_\lambda)\,\Tr\left[\Pi_\lambda \,\ketbra{\psi}\right]
\end{align}
(where $\Pi_\lambda$ is the projector onto $\mathcal{H}_\lambda$, $d_\lambda$ is its dimension), then $\ket{\psi}$ must have a weight $\Tr [\Pi_\lambda \ketbra{\psi}]$ on each $\mathcal{H}_\lambda$ that is proportional to $d_\lambda$—in particular, if all $\mathcal{H}_\lambda$'s are of the same dimension (e.g., one-dimensional), $\ket{\psi}$ must have an equal weight on all of them. 

For completeness, we mention that the situation becomes more subtle if some irreducible representations appear with multiplicity, that is, if the same (up to unitary equivalence) representation occurs more than once in the decomposition into irreducible subspaces. In such cases, there is in general no canonical way to distribute the fiducial state across the degenerate subspaces, and additional symmetry or optimization may be required to construct a valid measurement~\cite{Holevo1982,Chiribella2004a,Chiribella2006}. We will however not encounter such situations in the present paper: all irreducible subspaces will appear exactly once, and will be 1-dimensional.

\section{Global vs. Local Encoding}

The group-covariant constructions discussed above are a powerful way to construct symmetric measurements. In multipartite quantum systems, however, one is often interested in constructing measurements that respect a \emph{local} symmetry (e.g. tetrahedral symmetry in the case of the EJM). Specifically, for a system of \( n \) parties, we consider representations of the form
\begin{equation} \label{eq:tensorform}
    U_G = U_{G_1} \otimes U_{G_2} \otimes \cdots \otimes U_{G_n},
\end{equation}
where each \( U_{G_i} \) acts only on the Hilbert space of the \( i^\text{th} \) subsystem. Such constructions define tensor product representations, and allow one to generate entire bases by applying local unitaries to a fixed fiducial state.

This motivates the notion of a locally encodable basis~\cite{Tanaka007,Pimpel2023}: a set of orthonormal states \( \{ \ket{\psi_j} \}_j \) is locally encodable if there exists a single fiducial state \( \ket{\psi} \in \mathcal{H}_1 \otimes \cdots \otimes \mathcal{H}_n \) and local unitaries \( U_j^{(1)}, \dots, U_j^{(n)} \) such that
\begin{equation}
    \ket{\psi_j} = \left(U_j^{(1)} \otimes \cdots \otimes U_j^{(n)}\right) \ket{\psi}, \quad \forall\,j.
\end{equation}
In this sense, classical labels can be mapped to quantum states using only local control, without requiring global entangling operations. 

Tanaka \emph{et al.}~\cite{Tanaka007} studied this property in the context of quantum encoding and showed that various classes of entangled states—such as Dicke states and stabilizer states—are locally encodable via Pauli operations. More recently, Ref.~\cite{Pimpel2023} investigated a related question: when can a given multipartite state generate an orthonormal basis in which all elements are locally unitarily equivalent to the seed (and thus in particular share the same entanglement structure). While such bases always exist in the bipartite case, they conjecture that this is no longer true for systems of three or more parties. 
We emphasize that local encodability, in general, may still require global classical coordination to implement the appropriate local unitaries across parties. However, in the special case where the encoding unitaries can be chosen independently by each party—i.e., without any coordination—the fiducial state is known to be locally maximally entanglable (LME)~\cite{Kraus2009}.

Importantly, in any locally encodable basis the entire basis automatically inherits the entanglement structure of the fiducial state. In this sense, such bases are \emph{isoentangled} by construction. The group-theoretic framework we develop here provides a systematic method for generating large, structured families of isoentangled, locally encodable bases beyond those previously known.

\section{The Elegant Joint Measurement as a Regular Tetrahedral Basis}
\label{sec:EJM_as_reg_tetra}

Before proceeding to more general constructions, we analyze the EJM from the perspective outlined above—namely, as a concrete example of a two-qubit measurement whose elements form a group orbit under a tensor product representation. This instructive example will serve as a guiding case throughout the remainder of the paper.

\subsection{The standard EJM}
\label{subsec:standard_EJM}

The 2-qubit EJM basis $\{\ket{\psi_{\mathrm{EJM},j}}\}_{j=1}^4$ can be constructed from anti-parallel vectors in the Bloch sphere,
\begin{equation}
    \ket{\psi_{\mathrm{EJM},j}}= \frac{\sqrt{3}+1}{2\sqrt{2}}\ket{\vec{m}_j,-\vec{m}_j}+\frac{\sqrt{3}-1}{2\sqrt{2}}\ket{-\vec{m}_j,\vec{m}_j}, \label{eq:def_EJM_states}
\end{equation}
where $\ket{\vec{m}}$ denotes the eigenket of a qubit state with Bloch vector $\vec{m}$, and the $\vec{m}_j$s point to the vertices of a regular tetrahedron,
\begin{align}
    \vec{m}_1 &= (+1,+1,+1)/\sqrt{3}\,, \nonumber \\
    \vec{m}_2 &= (+1,-1,-1)/\sqrt{3}\,, \nonumber \\
    \vec{m}_3 &= (-1,+1,-1)/\sqrt{3}\,, \nonumber \\
    \vec{m}_4 &= (-1,-1,+1)/\sqrt{3}\,.
\end{align}

Since all four states in the EJM basis have identical Schmidt coefficients, they share the same amount of bipartite entanglement and thus form an \emph{isoentangled basis}. For bipartite pure states, the Schmidt coefficients completely characterize entanglement and are invariant under local unitaries. Consequently, it follows that each state \( \ket{\psi_{\mathrm{EJM},j}} \) can be obtained from a fixed fiducial state \( \ket{\psi_{\mathrm{EJM}}} \)---which we may take to be \( \ket{\psi_{\mathrm{EJM},1}} \)---via local unitaries:
\begin{equation}
    \ket{\psi_{\mathrm{EJM},j}} = U_j^{(1)} \otimes U_j^{(2)} \ket{\psi_{\mathrm{EJM}}},
\end{equation}
for some single-qubit unitaries \( U_j^{(1)}, U_j^{(2)} \). While these unitaries are not unique, the existence of such a decomposition follows directly from the shared Schmidt spectrum.

Moreover, the reduced single-qubit states associated with the EJM basis elements (i.e., the marginals obtained by tracing out one qubit) have Bloch vectors pointing to the four vertices of the regular tetrahedron formed by the $\vec{m}_j$s. This geometric arrangement is highly symmetric: the Pauli group, through its action on the Bloch sphere, permutes the vertices of this tetrahedron. In particular, the group 
\begin{equation}
  G_{\mathrm{tetra}}^{(2)}
  \;=\;
  \bigl\langle
      Z \otimes Z,\;
      X \otimes X
  \bigr\rangle
  \;\cong\;
  \mathbb{Z}_2^2, \label{eq:G_EJM}
\end{equation}
where $Z,X$ (and later, $Y$) are the Pauli operators, acts transitively on the basis states, generating the full EJM as a group orbit of \( \ket{\psi_{\mathrm{EJM}}} \). Since the generators \( Z \otimes Z \) and \( X \otimes X \) commute, the group is abelian and all elements are simultaneously diagonalizable. The resulting symmetry at the global level reflects a tetrahedral structure at the local level, as the action of each group element corresponds to a rotation of the Bloch vectors of the reduced states.

\subsection{Two-qubit tetrahedral bases}
\label{subsec:tetrahedral_bases}

The above discussion motivates a natural generalization. We define a two-qubit \emph{tetrahedral basis}  as any orthonormal basis obtained from the group orbit of a two-qubit fiducial state under \( G_{\mathrm{tetra}}^{(2)} \). That is, a two-qubit basis is a tetrahedral basis if and only if it can be written as
\begin{equation}
    \left\{ U_g \ket{\psi} \;\middle|\; U_g \in G_{\mathrm{tetra}}^{(2)} \right\},
\end{equation}
for some fiducial state \( \ket{\psi} \in \mathbb{C}^2 \otimes \mathbb{C}^2 \), provided this set forms an orthonormal basis.

To ensure this last constraint, the fiducial state must satisfy the completeness condition:
\begin{equation} \label{eq:completeness-EJM}
    \sum_{U_g \in G_{\mathrm{tetra}}^{(2)}} U_g \ketbra{\psi}{\psi} U_g^\dagger = \mathds{1} \,.
\end{equation}
This averaging corresponds to a special case of the \emph{Weyl twirling channel}~\cite{Popp2024}, which projects onto the subspace of operators that are diagonal in the Bell basis. Hence, the only states satisfying Eq.~\eqref{eq:completeness-EJM} are those whose group orbit lies entirely within this subspace.

It follows that any equal-weight superposition of Bell states (with arbitrary relative phases) generates a valid tetrahedral basis. Explicitly, any state of the form
\begin{equation} \label{eq:bell-superposition}
    \ket{\psi_{\mathrm{tetra}}^{(2)}} = \frac{1}{2} \sum_{z,x \in \{0,1\}} e^{i \alpha_{z,x}} \ket{\Phi_{z,x}}
\end{equation}
satisfies the completeness requirement, where \( \alpha_{z,x} \in \mathbb{R} \) and \( \ket{\Phi_{z,x}} \) denote the Bell states
\begin{align}
    \ket{\Phi_{0,0}} & = \tfrac{1}{\sqrt{2}}\big(\!\ket{0,0}{+}\ket{1,1}\!\big), & 
    \!\!\!\ket{\Phi_{0,1}} & = \tfrac{1}{\sqrt{2}}\big(\!\ket{0,0}{-}\ket{1,1}\!\big), \notag \\ 
    \ket{\Phi_{1,0}} & = \tfrac{1}{\sqrt{2}}\big(\!\ket{0,1}{+}\ket{1,0}\!\big), & 
    \!\!\!\ket{\Phi_{1,1}} & = \tfrac{1}{\sqrt{2}}\big(\!-\ket{0,1}+\ket{1,0}\!\big),
    \label{eq:Bellbasis}
\end{align}%
or 
\begin{equation} \label{eq:BellFromFourier}
     \ket{\Phi_{z,x}} = \mathrm{CNOT}_{(2 \rightarrow 1)} \left( \ket{z}_Z \otimes \ket{x}_X \right), 
 \end{equation}
where $\ket{\cdot}_Z$ and $\ket{\cdot}_X$ denote eigenstates of the $Z$ and $X$ Pauli operators, respectively, and where $\mathrm{CNOT}_{(i \rightarrow j)}$ denotes a CNOT gate where qubit $i$ controls qubit $j$.

Fixing an arbitrary global phase, this defines a three-parameter family of admissible fiducial states. For instance, for the standard EJM defined through Eq.~\eqref{eq:def_EJM_states}, \( \ket{\psi_{\mathrm{EJM}}} = \ket{\psi_{\mathrm{EJM},1}} \) corresponds (up to the choice of the global phase) to
\begin{equation} \label{eq:EJM_fid}
    \ket{\psi_{\mathrm{EJM}}} = \frac{1}{2} \ket{\Phi_{0,0}} + \frac{i}{2} \ket{\Phi_{0,1}} - \frac{i}{2} \ket{\Phi_{1,0}} + \frac{i}{2} \ket{\Phi_{1,1}},
\end{equation}
i.e. with the phases \( \alpha_{0,0} = 0 \), \( \alpha_{0,1} = \tfrac{\pi}{2} \), \( \alpha_{1,0} = -\tfrac{\pi}{2} \), and \( \alpha_{1,1} = \tfrac{\pi}{2} \) in Eq.~\eqref{eq:bell-superposition}.

To make this structure fully transparent—especially with an eye toward generalization—it is useful to take a slightly more formal perspective. Since \( G_{\mathrm{tetra}}^{(2)} \cong \mathbb{Z}_2^2 \) is abelian, all its elements can be simultaneously diagonalized, and it admits a complete set of orthonormal eigenstates. 
This joint eigenbasis is precisely the Bell basis of Eq.~\eqref{eq:Bellbasis}, which may be obtained as in Eq.~\eqref{eq:BellFromFourier}, by mapping a simpler representation of $\mathbb{Z}_2^2$ to $G_{\mathrm{tetra}}^{(2)}$ via an appropriate intertwiner (see Appendix~\ref{app:eigenbases} for details).

This highlights why equal superpositions of Bell states emerge as the admissible fiducial states in the tetrahedral family: they correspond to uniform superpositions over eigenstates of an abelian group representation, just as in standard Fourier constructions~\cite{Gross2006}. This will serve as the basis for generalizations to multipartite or higher-dimensional settings.

Before moving to these generalizations, however, let us look more closely at the bases constructed from the fiducial state of Eq.~\eqref{eq:bell-superposition}. Considering the reduced states of a single qubit, by construction the four Bloch vectors obtained from the four basis elements form a tetrahedron inscribed in the Bloch sphere. The specific shape of each local tetrahedron depends on the relative phases in the fiducial state. More precisely, we can compute the Bloch vectors of qubits 1 and 2 for the state of Eq.~\eqref{eq:bell-superposition}, as
\begin{equation}
    \vec r_1 = \tfrac{1}{2}\big(\vec r_+ + \vec r_-\big) , \quad \vec r_2 = \tfrac{1}{2}\big(\vec r_+ - \vec r_-\big)
\end{equation}
with
\begin{equation}
  \vec r_+ =
  \begin{pmatrix}
    \cos(\alpha_{1,0} - \alpha_{0,0}) \\[4pt]
    \sin(\alpha_{1,0} - \alpha_{0,1}) \\[4pt]
    \cos(\alpha_{0,1} - \alpha_{0,0})
  \end{pmatrix}\!, \ 
  \vec r_- =
  \begin{pmatrix}
    \cos(\alpha_{1,1} - \alpha_{0,1}) \\[4pt]
    \sin(\alpha_{1,1} - \alpha_{0,0}) \\[4pt]
    -\cos(\alpha_{1,1} - \alpha_{1,0})
  \end{pmatrix}\!.
\end{equation}
The remaining vertices of the local tetrahedra are generated by conjugating \(\vec r_1\) or \(\vec r_2\) with the Pauli reflections \(X, Y, Z\). Since reflections preserve distances and angles, the resulting tetrahedra have opposite edge pairs of equal length and are known as disphenoids~\cite{coxeter1973}. While the shapes of the tetrahedra on qubits 1 and 2 may differ in general, their Bloch vectors satisfy \(\|\vec{r}_1\| = \|\vec{r}_2\|\), so both tetrahedra are inscribed in spheres of equal radius (i.e., they share the same circumradius).
Interestingly one may note that they also have the same volume, $V\big(\vec r_i=(x_i,y_i,z_i)\big)=\frac{1}{6}\big|\det\!\big[(0,2y_i,2z_i),(2x_i,0,2z_i),(2x_i,2y_i,0)\big]\big|$ $=\frac{8}{3}|x_iy_iz_i|$ (indeed it can be verified that for $\vec r_1$ and $\vec r_2$ above, $x_1y_1z_1 +x_2y_2z_2=0$).

\subsection{Regular tetrahedral bases}

Among the above large family of two-qubit tetrahedral bases, the EJM corresponds to a special case in which the local Bloch vectors define regular tetrahedra. We now ask under which conditions the fiducial state leads to such regular local structures.

A natural first attempt is to align the Bloch vectors on both qubits,
\begin{equation} \label{eq:equal-bloch}
    \vec r_1 = \vec r_2.
\end{equation}
From the expressions \( \vec r_1 = \frac12(\vec r_+ + \vec r_-) \) and \( \vec r_2 = \frac12(\vec r_+ - \vec r_-) \), this condition requires \( \vec r_- = 0 \). However, one can see that \( \vec r_- = 0 \) in turn implies \( \vec r_+ = 0 \): indeed, setting for instance \( \alpha_{1,1} = 0 \), \( \vec r_- = 0 \) implies $\alpha_{0,0} = 0$ or $\pi$, $\alpha_{0,1} = \pm\frac{\pi}{2}$ and $\alpha_{1,0} = \pm\frac{\pi}{2}$, which imply $\cos(\alpha_{1,0} - \alpha_{0,0}) = \sin(\alpha_{1,0} - \alpha_{0,1}) = \cos(\alpha_{0,1} - \alpha_{0,0}) = 0$. Hence, the only possible solution for \( \vec r_1 = \vec r_2 \) is to have both Bloch vectors of length 0—i.e., no (nontrivial) regular tetrahedron arises.

By contrast, imposing instead that the Bloch vectors are  opposite,
\begin{equation} \label{eq:opposite-bloch}
    \vec r_1 = -\vec r_2,
\end{equation}
requires \( \vec r_+ = 0 \), but now allows for \( \vec r_- \ne 0 \). Setting for instance $\alpha_{0,0} = 0$, $\vec r_+ = 0$ implies $\alpha_{0,1} = \pm\frac{\pi}{2}$ and $\alpha_{1,0} = \pm\frac{\pi}{2}$, which then gives $\vec r_1 = -\vec r_2 = \frac12\vec r_- = \frac{\sin\alpha_{1,1}}{2}(\pm1,1,\pm1)$. Hence we get Bloch vectors that point along a body diagonal of the Bloch cube (i.e., in some direction \( (\pm 1, \pm 1, \pm 1) \))---from which, acting with the local Pauli operators \( X, Y, Z \) generates the vertices of a regular tetrahedron.
Thus, this anti-alignment condition is sufficient to generate mirrored regular tetrahedra on both qubits.

For the particular choice of phases
\begin{equation}
\label{eq:cyril-family-phases}
\alpha_{0,0} = 0, \ \
\alpha_{0,1} = \tfrac{\pi}{2}, \ \
\alpha_{1,0} = -\tfrac{\pi}{2}, \ \
\alpha_{1,1} = \theta + \tfrac{\pi}{2},
\end{equation}
 one gets the fiducial state
\begin{equation}
\label{eq:cyril-family}
\ket{\psi_{\rm EJM}^\theta} = \frac{1}{2}
\begin{bmatrix}
e^{i \frac{\pi}{4}} \\
-i \frac{1 + e^{i\theta}}{\sqrt{2}} \\
-i \frac{1 - e^{i\theta}}{\sqrt{2}} \\
e^{-i \frac{\pi}{4}}
\end{bmatrix},
\end{equation}
for which $\vec r_1 = -\vec r_2 = \frac{\cos\theta}{2}(1,1,1)$ (of length $\frac{\sqrt{3}}{2}\cos\theta$). This fiducial state generates the one-parameter family introduced in Ref.~\cite{Tavakoli2021}, that interpolates between the canonical EJM (for \( \theta = 0 \)) and the Bell basis (for \( \theta = \frac{\pi}{2} \), in which case the Bloch vectors vanish).

\section{Multi-qubit tetrahedral bases}
\label{sec:multiqubit}

As we saw, the group-theoretic reformulation of the EJM leads rather naturally to a broader class of two-qubit bases with tetrahedral symmetry. In particular, by identifying the EJM as the orbit of a fiducial state under an abelian subgroup of local unitaries, we have abstracted its structure in a way that generalizes to larger systems.

We now consider the extension of this construction to \( n \)-qubit systems. We require that the single-qubit marginals of each basis state exhibit local tetrahedral symmetry—replicating the Bloch vector structure of the EJM on each qubit.

This motivates the definition of \emph{multi-qubit tetrahedral bases}: orthonormal bases formed as the orbit of a fiducial state under a group \( G \subset \mathrm{SU}(2)^{\otimes n} \), where the local action of each group element preserves the tetrahedral configuration on every qubit.

To ensure the completeness condition is satisfied, we will focus on abelian groups \( G \), for which Schur’s lemma guarantees that a uniform superposition over the orbit yields the identity operator. This leads to two design criteria for the group \( G \):

\begin{enumerate}
    \item[(i)] \textbf{Local Pauli support:} Each tensor factor must consist of single-qubit Pauli operators only (\(\{\mathds{1}, X, Y, Z\}\)), and all Paulis must appear on each qubit. This ensures that, if the fiducial state has nonzero components along all Pauli directions, the group orbit generates a full tetrahedral pattern in each local Bloch sphere, where for $n\ge 3$, the same Bloch vectors will be obtained from several basis states. If a Bloch vector is confined to a plane for which one of the $X,Y,Z$ components is null (e.g., the \(X\!-\!Z\) plane), the local action still yields rectangular symmetry within that plane. If a Bloch vector is aligned with one of the $X,Y,Z$ directions, then the local action gets restricted to that single direction.
    
    \item[(ii)] \textbf{Abelian tensor structure:} The full representation of \(G\) should be abelian—ideally isomorphic to \(\mathbb{Z}_2^n\). This guarantees that all group elements are simultaneously diagonalizable, and any uniform superposition over a joint eigenbasis yields a valid orthonormal basis.
\end{enumerate}

These criteria naturally lead us to consider abelian subgroups of the \(n\)-qubit Pauli group, which are well studied in the context of stabilizer codes~\cite{Gottesman1997}. However, our goal here is different: instead of constructing states that are \emph{invariant} under a stabilizer group, we seek states that are \emph{covariant}—that is, states whose orbit under \(G\) spans a complete orthonormal basis.

There are exponentially many commuting \( n \)-tuples of Pauli operators that generate groups isomorphic to \( \mathbb{Z}_2^n \). Among these, we select one convenient representative that satisfies both criteria (i) and (ii). Our choice is not unique—any other group related by local Clifford unitaries or single-qubit Pauli rotations would serve equally well—but it is the simplest construction that fulfills our requirements.

Let \( Z^{(i)} \) and \( X^{(i)} \) denote the Pauli operators acting on qubit \( i \). We define the group
\begin{equation} \label{eq:ourSubgroup}
\begin{aligned}
G_{\mathrm{tetra}}^{(n)}
\;&=\;
\left\langle
Z^{(1)} Z^{(2)},
Z^{(2)} Z^{(3)},
\ldots,
Z^{(n-1)} Z^{(n)},
X^{\otimes n}
\right\rangle \\
&\cong\;
\mathbb{Z}_2^n,
\end{aligned}
\end{equation}
which consists of \( 2^n \) commuting \( n \)-qubit Pauli strings and meets both of the design criteria stated above.

The group $G_{\mathrm{tetra}}^{(n)}$ is abelian. We present in Appendix~\ref{app:eigenbases} a possible approach to obtain the eigenbasis in which all its elements can be diagonalized, showing that their eigenstates can be written in the form \cite{Toth2005}
\begin{equation}
\label{eq:FourierTetra}
    \ket{\Phi_{\vec{z}, x}}
    = S(n)\;
    \Bigg[ \bigotimes_{i=1}^{n-1} \ket{z_i}_Z \otimes \ket{x}_X \Bigg], 
\end{equation}
for all $\vec z=(z_1,\ldots,z_{n-1})\in\{0,1\}^{n-1}, x\in\{0,1\}$,  where we introduced the operator
\begin{equation}
\label{eq:CNOTchain}
S(n) = \mathrm{CNOT}_{(2 \to 1)}\cdots\mathrm{CNOT}_{(n \to n-1)}\,.
\end{equation}
More explicitly, we find that the eigenstates $\ket{\Phi_{\vec{z}, x}}$ are of the GHZ form $\frac{1}{\sqrt{2}}(\ket{j_1,\ldots,j_n}\pm\ket{j_1\oplus1,\ldots,j_n\oplus1})$ (where $\oplus$ denotes addition modulo 2).

As in the previous bipartite case, any \emph{equal-weight} superposition of these eigenvectors---i.e., any state of the general form \begin{equation}
\label{eq:fiducialCompact}
  \ket{\psi_{\mathrm{tetra}}^{(n)}}
  = \frac{1}{\sqrt{2^n}}
  \sum_{(\vec{z},x) \in \{0,1\}^n} 
      e^{i \alpha_{\vec{z}, x}}
      \ket{\Phi_{\vec{z}, x}},
\end{equation}
with \( \alpha_{\vec{z}, x} \in \mathbb{R} \)
defines a valid fiducial state.

Acting with the group elements of \( G_{\mathrm{tetra}}^{(n)} \) on this fiducial state produces a full multi-qubit tetrahedral basis. The completeness relation
\begin{equation}
  \sum_{U_g \in G_{\mathrm{tetra}}^{(n)}} U_g\, \ket{\psi_{\mathrm{tetra}}^{(n)}}\!\bra{\psi_{\mathrm{tetra}}^{(n)}}\, U_g^\dagger
  \;=\;
  \mathds{1}
\end{equation}
then follows directly from Schur’s lemma.

\subsection{Party-Permutation-Invariant fiducial states}
\label{subsec:PPI}

A natural question is whether the family of \(n\)-qubit tetrahedral bases admits a form of
party-permutation symmetry. Concretely, we ask for fiducial states
that are invariant under all permutations of the qubit subsystems:
\begin{equation}
\label{eq:perm-sym-def}
U_\sigma \ket{\psi_{\mathrm{PPI}}^{(n)}} = \ket{\psi_{\mathrm{PPI}}^{(n)}} \quad
\forall\,\sigma \in S_n,
\end{equation}
where \(U_\sigma\) denotes the unitary that permutes the tensor factors according to \(\sigma\in S_n\) (and with $S_n$ denoting the symmetric group). Note that while party-permutation symmetry of the fiducial does not imply that each basis state is itself party-permutation symmetric, all qubits share the same set of possible single-qubit marginals across the basis. Moreover, imposing party-permutation symmetry on the fiducial state greatly reduces the number of free parameters needed to specify it (and hence the entire basis). 

Since any fully symmetric state lies in the span of Dicke states, we consider fiducial states of the form
\begin{equation}
\label{eq:PPI-Dicke}
\ket{\psi_{\mathrm{PPI}}^{(n)}} = \sum_{k=0}^{n} a_k\, e^{i \alpha_k} \ket{D_k} \,,
\end{equation}
where \(\ket{D_k} := \sum_{\vec z:\mathrm{wt}(\vec z)=k} \ket{\vec z}\) denotes the supernormalized Dicke state of weight \(k\) (with $\mathrm{wt}(\cdot)$ denoting the Hamming weight), such that \(\norm{\ket{D_k}}^2 = \binom{n}{k}\), the \(a_k \in \mathbb{R}\) are amplitudes, and \(\alpha_k \in \mathbb{R}\) are phases.

We analyze in Appendix~\ref{app:PPI} the conditions under which the $G^{(n)}_{\mathrm{tetra}}$-orbit of $\ket{\psi_{\mathrm{PPI}}^{(n)}}$ forms an orthonormal basis.
We prove that solutions exist only for odd \(n\), in which case the amplitudes and phases can most generally be taken of the form
\begin{align}
& a_k = \frac{\cos \theta_k}{2^{(n-1)/2}}, \quad a_{n-k} = \frac{\sin \theta_k}{2^{(n-1)/2}}, \quad \alpha_{n-k} = \alpha_k + \frac{\pi}{2}, \notag \\[1mm]
& \hspace{42mm} \text{for } k=0,\dots,\tfrac{n-1}{2}, \label{eq:psiPPI_amplitudes_phases}
\end{align}
for some real parameter $\theta_k$. No such state exists for even \(n\).

We note that the requirement of Eq.~\eqref{eq:perm-sym-def} is rather strong.  One could relax it by demanding invariance only up to a global phase, or even by requiring merely that all single-qubit marginals coincide.  We conjecture that, for even $n$, no solution exists even under such a relaxed assumption. 
For odd $n$, it remains an open question whether there are solutions to the relaxed condition that are genuinely different from the one constructed above.

We illustrate the three-qubit case explicitly in the next section.

\subsection{Regular geometric bases}
\label{subsec:regular_geoms}

Within the general family of tetrahedral measurement bases, some fiducial states give rise to particularly symmetric configurations of Bloch vectors. These include regular tetrahedra—where the Bloch vectors form equidistant points on the Bloch sphere—as well as rectangular arrangements lying in a plane. In this subsection, we explore such geometrically structured bases. We start by focusing on the three-qubit case.

\subsubsection{Three-qubit regular tetrahedral bases}

In the two-qubit EJM, the single-qubit reductions form a pair of regular tetrahedra that are reflections of each other through the origin. For three qubits, we can consider two distinct options: 
\begin{enumerate}
    \item Bases where each qubit's reduction lies on a regular tetrahedron with the \emph{same} orientation.
    \item Bases where the single-qubit reductions are split between two opposite regular tetrahedra, reflections of each other through the origin, as in the two-qubit EJM. That is, some qubits have reductions aligned with one tetrahedron, while others align with its reflected counterpart.
\end{enumerate}

\begin{figure*}
    \centering
    \includegraphics[width=0.8\textwidth]{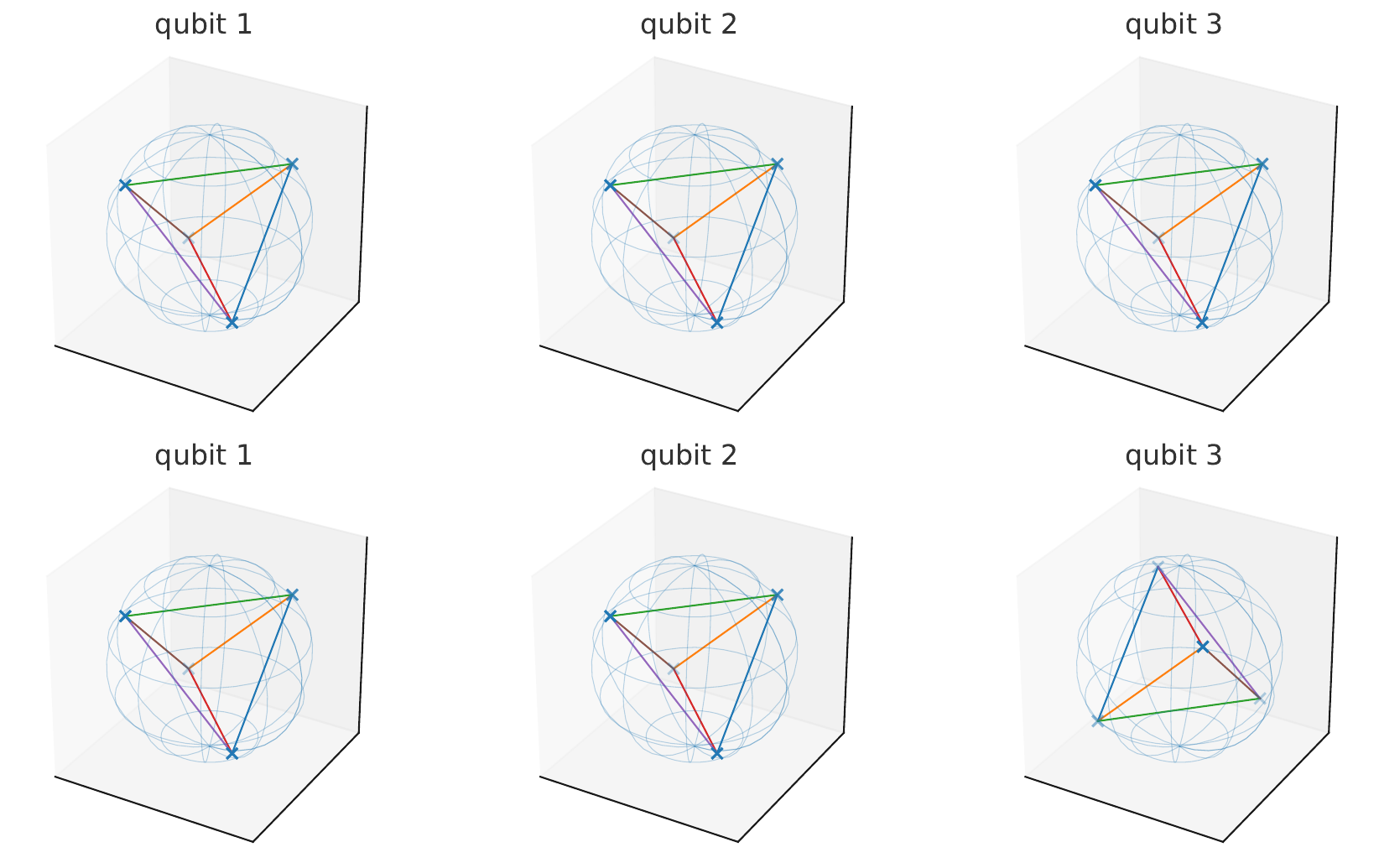}
    \caption{
 Configurations of local Bloch vectors generated by the three-qubit fiducial states of Eqs.~\eqref{eq:3qPPI} (top) and~\eqref{eq:ex_psi_both_orient} (bottom). Top: all three qubits share the same tetrahedral orientation. Bottom: qubits 1 and 2 share the same orientation, qubit 3 has the opposite orientation. Each configuration shows the four vertices of a regular tetrahedron inscribed in the Bloch sphere; for each qubit every Bloch vector occurs in two basis states and the pairing of states differs between qubits.}
    \label{fig:3qubitEJM}
\end{figure*}

\paragraph{Party-permutation invariant states}
We begin with the class of Pauli-encoded, party-permutation invariant (PPI) fiducial states. For three qubits, we have shown in the previous section that (for a fixed choice of a global phase) we can write
\begin{equation} \label{eq:3qubitPPI}
\ket{\psi_{\mathrm{PPI}}^{(3)}} = \frac{1}{2}
\begin{bmatrix}
\cos \theta_0 \\
e^{i\alpha_1} \,\cos \theta_1 \\
e^{i\alpha_1} \,\cos \theta_1 \\
i\,e^{i\alpha_1} \,\sin \theta_1 \\
e^{i\alpha_1} \,\cos \theta_1 \\
i\,e^{i\alpha_1} \,\sin \theta_1 \\
i\,e^{i\alpha_1} \,\sin \theta_1 \\
i\, \sin \theta_0
\end{bmatrix}.
\end{equation}
This state yields Bloch vectors of the form
\begin{equation}
  \vec r_{\mathrm{PPI}(3)} =
  \frac{1}{2}
  \begin{pmatrix}
    \cos (\alpha_1 ) \cos (\theta_0-\theta_1) \\[4pt]
    \sin (\alpha_1 ) \cos (\theta_0+\theta_1) + \sin (2\theta_1) \\[4pt]
    \cos (\theta_0+\theta_1) \cos (\theta_0-\theta_1)
  \end{pmatrix}.
\end{equation}

To ensure that the Bloch vector lies along the direction \( (1,1,1) \), so as to obtain a regular tetrahedron, we require all three components of $\vec r_{\mathrm{PPI}(3)}$ to be equal. This yields a nonlinear constraint on the angles \(\theta_0, \theta_1, \alpha_1\), which reduces to a degree-eight polynomial equation in their cosine. While not analytically solvable in general, numerical solutions can be found.

A simple solution is for instance
\begin{equation}
\ket{\psi} = \tfrac{1}{2}\big[
1,
\tfrac{9+3i}{10},
\tfrac{9+3i}{10},
\tfrac{-1+3i}{10},
\tfrac{9+3i}{10},
\tfrac{-1+3i}{10},
\tfrac{-1+3i}{10},
0
\big]^T, \label{eq:3qPPI}
\end{equation}
obtained for $\theta_0=0, \theta_1=\alpha_1=\arctan \frac{1}{3}$.
The Bloch structure of the corresponding basis is illustrated in Fig.~\ref{fig:3qubitEJM}~(top). The Bloch vector of each qubit reduced state is $(\frac{9}{20},\frac{9}{20},\frac{9}{20})$, of length \(  9\sqrt{3}/20 \approx 0.779 \). For comparison, we found numerically that for this configuration, the state with minimal entanglement has a local Bloch vector length of approximately \( 0.785 \), corresponding to slightly lower entanglement.

\paragraph{Reductions with both orientations}

Alternatively, we can construct states where the Bloch vectors  align along the two opposite tetrahedral orientations. For three qubits, this means two single-qubit Bloch vectors point toward \( (1,1,1) \) while the third one points toward \( (-1,-1,-1) \), or vice versa.

One such fiducial state is for instance
\begin{equation} \label{eq:ex_psi_both_orient}
   \ket{\psi} = \frac{1}{2\sqrt{10}} \big[3,-3-i,1,-3i,1,-3i,0,-i\big]^T,
\end{equation} 
which again yields Bloch vectors $\pm(\frac{9}{20},\frac{9}{20},\frac{9}{20})$, of length \(  9\sqrt{3}/20 \approx 0.779 \).  The first and second qubits point towards \( (1,1,1) \), while the third one points towards \( (-1,-1,-1) \). The Bloch vector structure of the corresponding basis is illustrated in Fig.~\ref{fig:3qubitEJM}~(bottom).

As before, one can numerically find solutions with slightly reduced entanglement, but these too involve high-degree polynomials that are not easily solved analytically.

\subsubsection{Planar, rectangular configurations}

A particularly simple and analytically tractable family of states arises when the fiducial state is restricted to the real subspace. In this case, the Bloch vectors lie entirely in the \( X-Z \) plane, and their arrangement under the Pauli group action forms a rectangle. 

The following real-valued fiducial state corresponds geometrically to a configuration of rectangles for any number of qubits:
\begin{equation}
\ket{\psi_\mathrm{rect}^{(n)}} =
\frac{1}{\!\sqrt{2^{n-1}}} \big[\underbrace{0,\, 0,\, \ldots,\, 0}_{2^{n-1}-1},\, 1 ,\, 0 ,\, \underbrace{1,\, 1,\, \ldots,\, 1}_{2^{n-1}-1} \big]^T.
\end{equation}
This is a real state supported only on computational basis vectors  of the form $\ket{0,1,1,\ldots,1}$ or $\ket{1,j_2,j_3,\ldots,j_n}$ with $(j_2,j_3,\ldots,j_n)\neq(0,0,\ldots,0)$. Each generator $Z^{(i)} Z^{(i+1)}$ flips the sign of a basis element if and only if the adjacent bit values at positions $i,i+1$ are different; over the support these $\pm1$ contributions cancel exactly (the single imbalance from excluding $\ket{1,0,\ldots, 0}$ is compensated by the vector $\ket{0,1,\ldots, 1}$), leading to zero expectation. Similarly, the global operator \(X^{\otimes n}\) maps every basis string to its bitwise complement, which lies outside the support. Hence all off-diagonal overlaps vanish, and the orbit defines a valid orthonormal basis.

For \( n = 2 \), the fiducial state is separable. For \( n > 2 \), the resulting states are partially entangled, with the degree of entanglement, measured by single-qubit entropy, decreasing with \(n\) (or equivalently, with local Bloch vector lengths increasing with $n$). 
Geometrically, for each qubit the four Pauli-conjugate points lie in the $X\!-\!Z$ plane and form a rectangle with side lengths $2\cdot\frac{4}{2^n}$ and $2\cdot\big(1-\frac{4}{2^n}\big)$; for the first qubit these two sides are interchanged with respect to the other qubits. This rectangle is a square for $n=3$.

This construction thus yields a simple $n$-qubit family of real-valued tetrahedral bases with rectangular symmetry in the local Bloch planes.

\section{Higher dimensions}
\label{sec:higher-dim}

A natural question is whether the group-theoretic construction underlying the EJM admits a meaningful extension to higher local dimensions. Ref.~\cite{Czartowski2021} proposed a generalization consisting of an orthonormal basis in \( \mathbb{C}^d \otimes \mathbb{C}^d \) that maximizes the single-party distinguishability of the basis vectors, meaning that the \( d^2 \) reduced density matrices on one subsystem are as far apart as possible in trace distance. Geometrically, this corresponds to the reductions forming a regular simplex in the space of density operators. When the reductions are pure, this condition defines a symmetric informationally complete measurement (SIC-POVM).
For \( d=2 \), the optimal basis coincides with the two-qubit EJM.

Explicitly, for every dimension $d$ where a SIC-POVM \( \{ \ket{\varphi_{j,k}} \}_{j,k=1}^{d} \) exists, Ref.~\cite{Czartowski2021} defined the two-qu$d$it basis
\begin{equation}
\label{eq:czartowski}
\ket{\psi^{(d)}_{j,k}} = \sqrt{\tfrac{d+1}{d}}\, \ket{\varphi_{j,k}} \otimes \ket{\varphi_{j,k}^\ast} - q\, e^{i\alpha} \ket{\Phi^+_d},
\end{equation}
where $\cdot^\ast$ denotes complex conjugation in the computational basis, and \( \ket{\Phi^+_d} = \frac{1}{\sqrt{d}} \sum_{j=0}^{d-1} \ket{j} \otimes \ket{j} \) is the maximally entangled state. 
To ensure that the states $\ket{\psi^{(d)}_{j,k}}$ define an orthonormal basis, the coefficients $q,\alpha$ must be related through
\begin{equation}
\label{eq:czartowski-coeffs}
    \cos\alpha =\frac{dq^2+1}{2q\sqrt{d+1}}
\end{equation}
(with $\frac{\sqrt{d+1}-1}{d}\leq q\leq \frac{\sqrt{d+1}+1}{d}$, $\sqrt{\frac{d}{d+1}}\le\cos\alpha\le1$); these together control the degree of entanglement.

Many known SIC-POVMs are covariant under the Weyl-Heisenberg (WH) group, meaning that each of its vectors can be written as
\begin{equation}
\ket{\varphi_{j,k}} = X_d^j Z_d^k \ket{\varphi_{\mathrm{fid}}},
\end{equation}
for some fixed fiducial state \( \ket{\varphi_{\mathrm{fid}}} \). Here, \( X_d \) and \( Z_d \) are the generalized shift and phase operators acting as
\begin{equation}
Z_d \ket{j} = \omega_d^j \ket{j}, \quad X_d \ket{j} = \ket{j\oplus 1}\,.
\end{equation}
where $\omega_d = e^{2i\pi/d}$ and $\oplus$ now denotes addition modulo $d$.
In the computational basis, they satisfy
\begin{equation}
Z_d^* = Z_d^{-1}, \qquad X_d^* = X_d,
\end{equation}
so complex conjugation inverts the phase operator but leaves the shift operator invariant.

It follows that each product state \( \ket{\varphi_{j,k}} \otimes \ket{\varphi_{j,k}^\ast} \) transforms covariantly under the group $\{U_g\otimes U_g^*|U_g\in WH\}$:
\begin{equation}
\ket{\varphi_{j,k}} \otimes \ket{\varphi_{j,k}^\ast}
= (X_d^j Z_d^k) \otimes (X_d^j Z_d^k)^* \left( \ket{\varphi_{\mathrm{fid}}} \otimes \ket{\varphi_{\mathrm{fid}}^\ast} \right).
\end{equation}
Moreover, \( \ket{\Phi^+_d} \) is invariant under \( U \otimes {U}^\ast \) for all unitaries \( U \), including WH operators. This implies that the full measurement basis can be written as a group orbit:
\begin{equation}
\label{eq:czartowski-orbit}
\ket{\psi^{(d)}_{j,k}} = (X_d^j Z_d^k) \otimes (X_d^j Z_d^k)^* \ket{\psi^{(d)}},
\end{equation}
where \( \ket{\psi^{(d)}} := \ket{\psi^{(d)}_{0,0}} \) from Eq.~\eqref{eq:czartowski}. Note that this fiducial state is not the SIC fiducial.

This representation makes it manifest that the measurement defined in Ref.~\cite{Czartowski2021} is covariant under the action of $\{U_g\otimes U_g^*|U_g\in WH\}$. As such, it fits naturally within our group-theoretic framework: the entire basis \( \{ \ket{\psi^{(d)}_{j,k}} \} \) is obtained as a group orbit of a fiducial state under an abelian symmetry group. This mirrors the EJM construction for qubits and motivates a broader generalization to any (uniform) local dimension.

Generalizing $G_{\mathrm{tetra}}^{(n)}$ from Eq.~\eqref{eq:ourSubgroup}, we thus define the $n$-qu$d$it abelian Pauli subgroup
\begin{equation} \label{eq:qudit-subgroup-n}
\begin{aligned}
G_d^{(n)}
\;&=\;
\left\langle
Z_d^{(1)} {Z_d^{(2)}}^*,
Z_d^{(2)} {Z_d^{(3)}}^*,
\ldots,
Z_d^{(n-1)} {Z_d^{(n)}}^*,
X_d^{\otimes n}
\right\rangle \\
&\cong\;
\mathbb{Z}_d^n.
\end{aligned}
\end{equation}
To construct the joint eigenbasis of its elements we follow a similar approach as in the qubit case, see Appendix~\ref{app:eigenbases}. We obtain eigenstates of the form 
\begin{equation}
\label{eq:Fourier-qudit-final}
    \ket{\Phi_{\vec{z}, x}^{(d)}}
    = S_d(n)\;
    \Bigg[ \bigotimes_{i=1}^{n-1} \ket{z_i}_{Z_d} \otimes \ket{x}_{X_d} \Bigg], 
\end{equation}
for all $\vec z\in\mathbb{Z}_d^{n-1}, x\in\mathbb{Z}_d$, where $\ket{\cdot}_{Z_d}$ and $\ket{\cdot}_{X_d}$ denote eigenstates of the $Z_d$ and $X_d$ operators, respectively, and with now the operator
\begin{equation}
\label{eq:SUMchain}
S_d(n) = \mathrm{SUM}^{(d)}_{(2 \to 1)}\cdots\mathrm{SUM}^{(d)}_{(n \to n-1)}
\end{equation}
in which each $\mathrm{SUM}^{(d)}_{(i+1 \to i)}$ is a SUM gate acting as \(\mathrm{SUM}_{(i+1 \to i)} \ket{j,k} = \ket{j\oplus k, k}\) on qudits $i$ and $i+1$. The eigenstates $\ket{\Phi_{\vec{z}, x}^{(d)}}$ are found to be higher-dimensional Bell- or GHZ-type of states.

Any equal-weight superposition of these vectors then defines a fiducial state whose group orbit spans a complete orthonormal basis:
\begin{equation}
\label{eq:genFid}
\ket{\psi_d^{(n)}} =
\frac{1}{\sqrt{d^n}} \sum_{\vec{z},\,x}
e^{i \alpha_{\vec{z},x}} \ket{\Phi^{(d)}_{\vec{z},x}}.
\end{equation}
This defines a large family of entangled bases with local WH invariance which includes the construction of Ref.~\cite{Czartowski2021} as a special case.

For instance, for $d=3$ (and coming back to the $n=2$ case), writing the Hesse SIC fiducial state~\cite{Dang2013} 
\begin{equation}
\ket{\varphi_{\mathrm{SIC}}} = \frac{1}{\sqrt{2}} \begin{bmatrix} 1 \\ -e^{i \theta} \\ 0 \end{bmatrix},
\end{equation}
with $\theta$ a free real parameter
into Eq.~\eqref{eq:czartowski} with $q=1/3$ and $\alpha = 0$ (corresponding to the most distinguishable local marginals), the overlaps with the eigenbasis yield relative phase angles:
\begin{equation}
\left[\alpha_{z,x}\right] =
\begin{bmatrix}
\phantom{-}0        & \tfrac{\pi}{3} & -\tfrac{\pi}{3} \\[1mm]
\theta+\pi & \theta+\pi & \theta+\pi \\[1mm]
-\theta+\pi & -\theta-\tfrac{\pi}{3}         & -\theta+\tfrac{\pi}{3}
\end{bmatrix}.
\end{equation}
This recovers, as a special case of our general construction, the $d=3$ phase structure underlying the examples of Ref.~\cite{Czartowski2021}.

\section{Localizability}
\label{sec:Localizability}

\begin{figure}[ht!]
    \centering
    \includegraphics{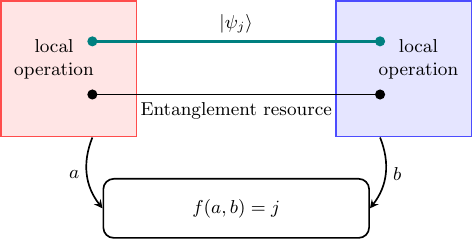}
    \caption{Illustration of localization of a joint measurement for two parties. 
    An unknown joint basis state \( \ket{\psi_j} \) is shared between two parties, who also possess a shared entanglement resource. Each party applies a local operation resulting into classical outcomes $a$ and $b$, which are used to deterministically compute the measurement result \( f(a,b) = j \). }
    \label{fig:localize}
\end{figure}

All bases introduced in this work are, by design, locally \emph{encodable}: each basis state can be prepared from a single fiducial state via local unitary transformations. However, the inverse task—decoding or ``measuring'' the classical label \footnote{Note that from the linearity of the Born rule, if one is able to distinguish the eigenstates, then one is also able to reproduce all statistics of the measurement} of an unknown basis state using only local operations—is generally impossible unless the basis only consists of product states. In such cases, distinguishing basis elements requires global operations~\cite{Popescu1994}. Nonetheless, if the decoding parties have access to pre-shared entanglement, the measurement may be implemented using only local quantum operations followed by broadcast communication—that is, in a scenario where only joint classical postprocessing of local outcomes is permitted~\cite{Vaidman2003,Groisman2003}, see Fig.~\ref{fig:localize}.

Ref.~\cite{Pauwels2025} introduced a framework that classifies projective measurements according to their \emph{localization cost}—the minimum amount of pre-shared entanglement needed to realize the measurement with local operations and classical postprocessing. The classification is organized into operational “levels’’ built from a finite-round adaptation of Vaidman’s scheme~\cite{Vaidman2003}: starting from the unknown joint state, one performs a fixed number of cyclic, back-and-forth blind teleportations (i.e., teleportations where the Pauli corrections are not communicated), and then a single joint measurement at the final receiver, with all parties classically combining their local data at the end. The number of such rounds defines the level in the classification: here we take level $k$ to correspond to $k-1$ rounds of back-and-forth blind teleportations (so that the lowest  level, $k=1$, corresponds to no round of teleportation at all)~\footnote{In~\cite{Pauwels2025} we used a convention where the base level was numbered~0. Here we shift all levels by one, so as to align with the usual Clifford-hierarchy convention.}. Each teleportation introduces an independent Pauli distortion known only to the party performing the local Bell state measurement; operationally, the protocol branches over these $d^2$ possibilities and ensures that, after $k-1$ rounds, the classical postprocessing recombines the branches to reproduce the target measurement statistics.
Formally, let $\mathcal{V}^{d,n}_k$ denote the set of $n$-party, $d$-dimensional projective measurements that admit such a $(k{-}1)$-round localization. This yields a nested hierarchy
\begin{equation}
\mathcal{V}^{d,n}_1 \subsetneq \cdots \subsetneq \mathcal{V}^{d,n}_k \subsetneq \cdots,
\end{equation}
where each level~$k$ is obtained from level~$k\,-\,1$ by allowing one additional blind back-and-forth teleportation cycle. Intuitively, the higher the level, the larger the class of Pauli-distortion patterns that can be re-expressed as valid correction operators from the previous level of the hierarchy and the more nonlocal the measurements one can implement with only local operations and joint classical postprocessing. Localization cost thus provides an operational measure of complexity, distinguishing between measurements that differ in their requirements for entanglement-assisted local realization. In terms of resources, an implementation at level~$k$ consumes a fixed amount of pre-shared entanglement that grows exponentially in~$k$ and super-exponentially in the number of parties.
In the resulting classification, the lowest level (no teleportation, no shared entanglement) contains only product measurements. For two qubits, the second level contains a single entangled measurement, the Bell-state measurement. Partially entangled bases—bases where the defining projectors are partially entangled—appear at higher levels. Remarkably, the Elegant Joint Measurement (EJM) is among the earliest such examples, occupying the third level, alongside only a handful of others~\cite{Pauwels2025}.

This observation raises a natural question: which of the more general tetrahedral measurements constructed in this work admit similarly efficient local implementations? To tackle this problem, let us start by recalling the explicit algebraic characterization of the $d^n\times d^n$ measurement unitaries $M \in \mathcal{U}(d^n)$ that are localizable at each level~$k$, which we will use it to classify the families constructed here.

\begin{proposition}[Level-$k$ localizable measurements]
A measurement basis \( M \) is localizable at level \( k \) if and only if \( M \in \mathcal{V}_k^{d,n} \), where
\begin{equation}
  \mathcal{V}_k^{d,n} := \left\{ M \in \mathcal{U}(d^n) \;\middle|\; M (\mathds{1} \otimes \mathcal{P}_{n-1}^d) M^\dagger \in \bar{\mathcal{V}}_{k-1}^{d,n} \right\},
  \label{eq:VaidmanHierarchy}
\end{equation}
and \( \mathcal{P}_n^d = \{X_d,Z_d\}^{\otimes n} \) is the \( n \)-qudit Weyl-Heisenberg group generated by the generalized shift and phase operators \( X_d \) and \( Z_d \). 

The auxiliary sets \( \bar{\mathcal{V}}_k^{d,n} \) are defined recursively as:
\begin{equation}
  \bar{\mathcal{V}}_k^{d,n} := \left\{ M \in \mathcal{U}(d^n) \;\middle|\; M \mathcal{P}_n^d M^\dagger \in \bar{\mathcal{V}}_{k-1}^{d,n} \right\},
  \label{eq:VaidmanHierarchy2}
\end{equation}
with base cases \( \bar{\mathcal{V}}_1^{d,n} := \mathcal{V}_1^{d,n} := \{ P \cdot \Phi \} \), the set of all permutation matrices $P$ with arbitrary phases $\Phi$.
\end{proposition}

Explicitly characterizing the sets \( \mathcal{V}_k^{d,n} \) is highly nontrivial. However, a sufficient (though not necessary) condition for localizability at level~\( k \)—satisfied, for instance, by the EJM—is that the measurement unitary \( M \in \mathcal{C}_k^{d,n} \), where \( \mathcal{C}_k^{d,n} \) denotes the \( k^\text{th} \) level of the Clifford hierarchy~\cite{Gottesman1999} (see also Appendix~\ref{app:clifford}), defined recursively via conjugation with respect to Weyl–Heisenberg operators:
\begin{equation}
\mathcal{C}_{k}^{d,n} := \left\{ U \in \mathcal{U}(d^n)\;\middle|\; U \mathcal{P}_n^d U^\dagger \subseteq \mathcal{C}_{k-1}^{d,n} \right\} \,,
\end{equation}
with the base case $\mathcal{C}_1^{d,n} := \mathcal{P}_n^d$.
Indeed, we always have the inclusion~\cite{Pauwels2025}
\begin{equation}
\mathcal{C}_k \subsetneq \mathcal{V}_k\,,
\end{equation}
where we dropped the superscripts $n$ and $d$ to simplify notation.
This inclusion is useful because the Clifford hierarchy is relatively well understood, allowing us to exploit its algebraic structure. 

In the case of Weyl-Heisenberg-encoded tetrahedral measurements, this connection enables a concrete analysis. The theorem below isolates the key structural feature that determines localizability for all such bases: the Clifford level of a single diagonal phase gate.

To identify the relevant phase gate, let us first note that from the general form of tetrahedral fiducials in Eq.~\eqref{eq:genFid}, and using the definition of the eigenstates $\ket{\Phi_{\vec{z}, x}^{(d)}}$ from Eq.~\eqref{eq:Fourier-qudit-final}, we can define the following normal form for the unitary that prepares any fiducial state:
\begin{observation}[Fiducial normal form]
Every tetrahedral fiducial state can be prepared  as $\ket{\psi} = V\ket{0}^{\otimes n}$ using a unitary of the form
\begin{equation} \label{eq:GenFidPhase}
V = S_d(n)\, {H_d^{(n)}}^\dagger\, D_{\vec \alpha}\, H_d^{\otimes n}\,,
\end{equation}
where $D_{\vec \alpha}$ is a diagonal phase gate (which contains the phases $e^{i \alpha_{\vec{z},x}}$ from Eq.~\eqref{eq:genFid}), $H_d =\sum_{j,k} \omega_d^{j k}\ketbra{j}{k}$ is the generalized $d$-dimensional Hadamard gate (or discrete Fourier transform), and $S_d(n)$ is the operator defined in Eq.~\eqref{eq:SUMchain}. 
\end{observation}

We then have the following theorem:
\begin{theorem}[Clifford level of a tetrahedral measurement unitary]\label{lem:localizability}
Fix a fiducial state $\ket{\psi}$ and define the basis-change unitary
\begin{equation}
M_\psi \;:=\; \sum_{g:U_g\in G^{(n)}_d} U_g\,\ket{\psi}\bra{g},
\end{equation}
with $\{\ket{g}\}$ denoting the computational basis, so that $M_\psi\ket g=U_g\ket{\psi}=\ket{\psi_g}$ (and with $U_0=\mathds{1}$).
Considering the diagonal phase gate $D_{\vec \alpha}$ that appears in the normal form of Eq.~\eqref{eq:GenFidPhase}, one has, for $k\ge2$, that
\begin{equation}\label{eq:fiducial_determines_Clifford}
D_{\vec \alpha}\in\mathcal{C}_k \quad\Longrightarrow\  M_\psi \in \mathcal{C}_k.
\end{equation}
\end{theorem}

\begin{proof}
Suppose $\ket{\psi}=S_d(n) {H_d^{(n)}}^{\dagger} D_{\vec\alpha} H_d^{\otimes n}\ket{0}^{\otimes n}$ with $D_{\vec\alpha}\in\mathcal{C}_k$, for some $k\ge 2$.
Define the fixed Clifford operator
\begin{equation}
C_\text{diag} := {H_d^{(n)}} S_d(n)^\dagger \ \ \in\mathcal{C}_2
\end{equation}
which, as shown in Appendix~\ref{app:eigenbases} (see Eq.~\eqref{eq:change_group_rep}), maps the generators of $G^{(n)}_d$ to the single-site (diagonal) operators $Z_d$. 
Conjugating $M_\psi$ on the left, let us then define:
\begin{equation}
\tilde{M}_\psi := C_\text{diag}\, M_\psi = \sum_g \tilde{U}_g (C_\text{diag} \ket{\psi}) \bra{g}\,
\end{equation}
with $\tilde{U}_g{:=}C_\text{diag} U_{\!g} C_\text{diag}^\dagger$.
Using $C_\text{diag} \ket{\psi}= D_{\vec \alpha} H_d^{\otimes n}\!\ket{0}^{\!\otimes n}$ and the fact that $D_{\vec \alpha}$ and $\tilde{U}_g$ are all diagonal and therefore commute, we may write 
\begin{equation}
\tilde{M}_\psi = D_{\vec \alpha} \qty[\sum_g \tilde{U}_g H_d^{\otimes n}\ket{0}^{\otimes n}\bra{g}] \,. 
\end{equation}
Finally, remembering that each $\tilde{U}_g$ is a product of $Z_d$ operators and using the identity $Z_d H_d = H_d X_d$, we can write
\begin{equation}
\sum_g \tilde U_g\,H_d^{\otimes n}\ket{0}^{\otimes n}\bra g
\;=\; \sum_g H_d^{\otimes n}\ket{g}\bra g
\;=\; H_d^{\otimes n},
\end{equation}
which is in $\mathcal{C}_2 \subseteq \mathcal{C}_k$.  Therefore $\tilde{M}_\psi=D_{\vec\alpha}\,H_d^{\otimes n}\in\mathcal{C}_k$, and finally
$M_\psi=C_{\mathrm{diag}}^\dagger \tilde{M} _\psi \in\mathcal{C}_k$.
\end{proof}

Remembering that $\mathcal{C}_k \subsetneq \mathcal{V}_k$, the theorem above implies that the localizability of any tetrahedral basis can be analyzed by examining the diagonal phase gate used to construct the fiducial. This is particularly advantageous, as the structure of such gates is fully characterized by the phase-polynomial formalism of Ref.~\cite{Gottesman2017} (see Appendix~\ref{app:clifford}). In what follows, we specialize to the qubit case (i.e., \( d = 2 \)). The higher-dimensional case remains open and may be addressed in future work.

Cui-Gottesman-Krishna's phase-polynomial framework~\cite{Gottesman2017} considers diagonal unitary operators on \(n\) qubits of the form
\begin{equation} 
D_{f_m} = \sum_{\vec{z} \in \mathbb{Z}_2^n}
    \exp\!\left[i\,\tfrac{2\pi}{2^m}\,f_m(\vec{z})\right]
    \,\bigl|\vec{z}\bigr\rangle\!\bigl\langle\vec{z}\bigr|, \label{eq:phasepoly}
\end{equation}
where the phase polynomial \( f_m : \mathbb{Z}_2^n \to \mathbb{Z}_{2^m} \) is defined as
\begin{equation}
f_m(\vec{z}) =
\sum_{S \subseteq [n] \,\;S \neq \emptyset} a_S \prod_{i \in S} z_i \mod 2^m,
\end{equation}
with coefficients \( a_S \in \mathbb{Z}_{2^m} \), and where \( m \) is referred to as the \emph{precision} of the polynomial. Any diagonal unitary operator can indeed be approximated by an operator of this form, up to a precision characterized by the value of $m$. 

For such operators, it was shown~\cite{Gottesman2017} that the lowest level $k$ for which $D_{f_m} \in \mathcal{C}_k$ is given by
\begin{equation}\label{eq:refined-level}
k \;=\; \max_{S \neq \emptyset : \,a_S \neq 0} \big[ (m - \nu_2(a_S) - 1) + |S| \big],
\end{equation}
where the maximization runs over all nonempty monomial sets $S$, where \( |S| \) denotes the cardinality of $S$ (i.e., the number of variables appearing in the corresponding monomial), and \( \nu_2(a_S) \) is the 2-adic valuation (i.e., the exponent of the highest power of 2 dividing \( a_S \in \mathbb{Z}_{2^m} \)).
In particular, whenever \( a_S \) is divisible by a power of 2, the effective precision of that term is reduced by \( \nu_2(a_S) \), lowering its contribution to the hierarchy level.
We refer to Appendix~\ref{app:clifford} for details. 

This characterization provides a concrete and efficient way to classify diagonal gates and thus to identify localizable tetrahedral measurements according to their symmetry and Clifford hierarchy level.

This criterion allows us to systematically enumerate all diagonal gates in \( \mathcal{C}_k \) by scanning over all coefficient tuples $\{a_S\}_S$ with each $a_S \in \mathbb{Z}_{2^m}$ such that the maximal contribution in Eq.~\eqref{eq:refined-level} equals the target level \( k \). Each such coefficient tuple defines a phase polynomial \( f_m \), and hence a diagonal unitary \( D_{f_m} \in \mathcal{C}_k \).

For each polynomial, we construct the associated tetrahedral basis via the fiducial state given by Eq.~\eqref{eq:GenFidPhase}, with the diagonal gate defined according to Eq.~\eqref{eq:phasepoly}. The resulting bases are classified geometrically.

\subsection{Two-qubit hierarchy}

We now systematically classify all tetrahedral two-qubit bases appearing in the second, third, and fourth levels of the Clifford hierarchy using Cui-Gottesman-Krishna's characterization.

For two qubits, the Boolean phase polynomial
\begin{equation}
  f_m \colon \mathbb{Z}_2^2 \to \mathbb{Z}_{2^m}
\end{equation} 
takes the general form
\begin{equation}
  f_m(z_1,z_2) = a_{1} \,z_1 + a_{2} \,z_2 + a_{1,2}\, z_1 \,z_2 \mod 2^m,
\end{equation} 
where \( a_S \in \mathbb{Z}_{2^m} \) are coefficients indexed by subsets \( S \subseteq \{1,2\} \), and \( m \) is the phase precision.

To determine the Clifford level of the associated diagonal unitary, we apply Cui-Gottesman-Krishna's criterion~\eqref{eq:refined-level}.
This classification is implemented by enumerating over the phase precision $m$. For each fixed $m$, we list all coefficient triples $a_S \in \mathbb{Z}_{2^m}^3$. Note that we can always assume one coefficient is odd; i.e., \(\exists\,S\) with \(\nu_2(a_S)=0\), otherwise this would lower the effective precision $m$. For every pair $(m,a_S)$, we compute the resulting level $k$ via Eq.~\eqref{eq:refined-level} and bin by $k$. It then suffices to consider \(m\in\{k-1,k\}\). Indeed, for two qubits one always has \(k\le m-1+n=m+1\). If \(k<m\), Eq.~\eqref{eq:refined-level} would require \(\nu_2(a_S)>|S|-1\ge 0\) for all \(S\), contradicting the existence of an odd coefficient. Hence \(k\ge m\), so only \(k\in\{m,m+1\}\) can occur. Equivalently, when reporting by level \(k\), the contributing instances necessarily come from \(m\in\{k-1,k\}\).
We then build the corresponding tetrahedral bases and classify according to their geometry. We now perform this analysis explicitly for levels \( k = 2, 3, 4 \). Associated computational tools are provided in Ref.~\cite{CompAppendix}.

At level \( k=1 \), only trivial, separable geometries appear. At level \( k=2 \) we get a pointlike configuration, corresponding to the maximally entangled (Bell) basis, corresponding to the polynomial \[f_{m=2}(z_1,z_2) = z_1+z_2 \mod 2^2\,.\]

At level \( k = 3 \), we identify exactly three distinct tetrahedral geometries, represented by the following polynomials:
\begin{enumerate}
  \item \textbf{Elegant Joint Measurement (EJM):}
    \[
      f_{m=2}(z_1,z_2) = z_1 z_2 \mod 2^2,
    \]
    yielding a regular tetrahedron with pairwise equidistant Bloch vectors each of length $\frac{\sqrt{3}}{2}$.

  \item \textbf{Rectangular configuration (\(1 \times \sqrt{2}\)):}
    \[
      f_{m=3}(z_1,z_2) = z_1 + z_2 \mod 2^3,
    \]
    corresponds to a rectangle in the Bloch sphere with Bloch vectors of size $\frac{\sqrt{3}}{2}$ (like the EJM), e.g. the first Bloch vector of the first qubit is $(\frac{1}{\sqrt{2}},\frac12,0)^T$. It is equivalent to the \( E_2 \) basis introduced in~\cite{Pauwels2025}.

  \item \textbf{Linear configuration (\( \sqrt{2} \)-line):}
    \[
      f_{m=3}(z_1,z_2) = z_1 + 2 z_2 \mod 2^3,
    \]
    produces a linear configuration with the first qubit Bloch vector corresponding to $(\frac{1}{\sqrt{2}},0,0)^T$, of length $\frac{1}{\sqrt{2}}$. This recovers the twisted Bell basis discussed in Refs.~\cite{Boreiri2023,Pauwels2025}.
\end{enumerate}
The set \( \mathcal{V}_3 \) contains five inequivalent two-qubit bases~\cite{Pauwels2025}, three of which we see appear here. This reflects an underlying group-theoretic structure: all three are generated by the same symmetry orbit. The two remaining bases from Ref.~\cite{Pauwels2025} do not appear here. The first—the partial Bell basis—is excluded as it is not isoentangled. The second, a distinct member of the Bell family recently discussed in~\cite{DelSanto2024}, lies within the broader class \( \mathcal{V}_3 \) but not \( \mathcal{C}_3 \).

\begin{table*}[ht!]
\centering
\begin{tabular}{|c|c|c|c|c|} 
 \hline
 $m$ & \begin{tabular}{c} Phase \\ polynomial \end{tabular} & Fiducial state & \begin{tabular}{c} Marginal states \\ (Bloch vectors) \end{tabular} & Geometry \\[1ex]
 \hline
 3 & $z_1+z_2+z_1\,z_2$ & $[e^{i\frac{\pi}{8}}\frac{c}{\sqrt{2}},\frac12,\frac{i}{2},e^{-i\frac{3\pi}{8}}\frac{s}{\sqrt{2}}]^T$ & \begin{tabular}{c} $[\frac{1}{2\sqrt{2}},\frac{1}{2\sqrt{2}},\frac{1}{2\sqrt{2}}]^T$ \\ $[\frac{1}{2\sqrt{2}},-\frac{1}{2\sqrt{2}},\frac{1}{2\sqrt{2}}]^T$ \end{tabular} & \begin{tabular}{c} Regular \\ tetrahedra \end{tabular} \\[1ex] 
 \hline
 3 & $z_1\,z_2$ & $[\frac{1}{\sqrt{2}},e^{-i\frac{3\pi}{8}}\frac{s}{\sqrt{2}},e^{i\frac{\pi}{8}}\frac{c}{\sqrt{2}},0]^T$ & \begin{tabular}{c} $[c^2,\frac{1}{2\sqrt{2}},s^2]^T$ \\ $[s^2,-\frac{1}{2\sqrt{2}},c^2]^T$ \end{tabular} & \begin{tabular}{c} Irregular \\ tetrahedra \end{tabular} \\[1ex]
 \hline
 4 & $z_1+z_2+2\,z_1\,z_2$ & $[\frac{1+c+is}{2\sqrt{2}},\frac{c+is-i}{2\sqrt{2}},\frac{c+is+i}{2\sqrt{2}},\frac{1-c-is}{2\sqrt{2}}]^T$ & \begin{tabular}{c} $[\frac{c}{\sqrt{2}},\frac{1}{2},\frac{s}{\sqrt{2}}]^T$ \\ $[\frac{s}{\sqrt{2}},-\frac{1}{2},\frac{c}{\sqrt{2}}]^T$ \end{tabular} & \begin{tabular}{c} Irregular \\ tetrahedra \end{tabular} \\[1ex] 
 \hline
 4 & $z_1+z_2+4\,z_1\,z_2$ & $[\frac{1+c+is}{2\sqrt{2}},e^{i\frac{\pi}{8}}\frac{1+s-ic}{2\sqrt{2}},e^{i\frac{\pi}{8}}\frac{1-s+ic}{2\sqrt{2}},\frac{1-c-is}{2\sqrt{2}}]^T$ & \begin{tabular}{c} $[\frac{s}{\sqrt{2}},\frac{1}{2\sqrt{2}},\frac{c}{\sqrt{2}}]^T$ \\ $[\frac{c}{\sqrt{2}},-\frac{1}{2\sqrt{2}},\frac{s}{\sqrt{2}}]^T$ \end{tabular} & \begin{tabular}{c} Irregular \\ tetrahedra \end{tabular} \\[1ex] 
 \hline
 4 & $z_1+z_2$ & $[\frac{1+c+is}{2\sqrt{2}},e^{i\frac{\pi}{8}}\frac{1-c-is}{2\sqrt{2}},e^{i\frac{\pi}{8}}\frac{1+c+is}{2\sqrt{2}},\frac{1-c-is}{2\sqrt{2}}]^T$ & \begin{tabular}{c} $[c,\frac{1}{2\sqrt{2}},0]^T$ \\ $[0,-\frac{1}{2\sqrt{2}},c]^T$ \end{tabular} & Rectangles \\[1ex] 
 \hline
 4 & $3\,z_1+3\,z_2$ & $[\frac{1+s+ic}{2\sqrt{2}},e^{i\frac{3\pi}{8}}\frac{1-s-ic}{2\sqrt{2}},e^{i\frac{3\pi}{8}}\frac{1+s+ic}{2\sqrt{2}},\frac{1-s-ic}{2\sqrt{2}}]^T$ & \begin{tabular}{c} $[s,\frac{1}{2\sqrt{2}},0]^T$ \\ $[0,-\frac{1}{2\sqrt{2}},s]^T$ \end{tabular} & Rectangles \\[1ex] 
 \hline
 4 & $z_1+4\,z_2$ & $[e^{i\frac{\pi}{4}}\frac12,e^{-i\frac{\pi}{8}}\frac12,e^{i\frac{3\pi}{8}}\frac12,e^{-i\frac{\pi}{4}}\frac12]^T$ & \begin{tabular}{c} $[c,0,0]^T$ \\ $[0,-c,0]^T$ \end{tabular} & Line segments \\[1ex] 
 \hline
 4 & $3\,z_1+4\,z_2$ & $[e^{i\frac{\pi}{4}}\frac12,e^{i\frac{\pi}{8}}\frac12,e^{i\frac{5\pi}{8}}\frac12,e^{-i\frac{\pi}{4}}\frac12]^T$ & \begin{tabular}{c} $[s,0,0]^T$ \\ $[0,-s,0]^T$ \end{tabular} & Line segments \\[1ex] 
 \hline
 4 & $z_1+2\,z_1\,z_2$ & $[\frac{1}{\sqrt{2}},e^{-i\frac{\pi}{4}}\frac{s}{\sqrt{2}},e^{i\frac{\pi}{4}}\frac{c}{\sqrt{2}},0]^T$ & \begin{tabular}{c} $[\frac{c}{\sqrt{2}},\frac{c}{\sqrt{2}},s^2]^T$ \\ $[\frac{s}{\sqrt{2}},-\frac{s}{\sqrt{2}},c^2]^T$ \end{tabular} & \begin{tabular}{c} Different \\ tetrahedra \end{tabular} \\[1ex] 
 \hline
 4 & $z_1+2\,z_2$ & $[e^{i\frac{\pi}{8}}\frac{c}{\sqrt{2}},e^{-i\frac{\pi}{4}}\frac{s}{\sqrt{2}},e^{i\frac{\pi}{4}}\frac{c}{\sqrt{2}},e^{-i\frac{3\pi}{8}}\frac{s}{\sqrt{2}}]^T$ & \begin{tabular}{c} $[c,\frac{s}{\sqrt{2}},0]^T$ \\ $[0,-\frac{c}{\sqrt{2}},\frac{1}{\sqrt{2}}]^T$ \end{tabular} & \begin{tabular}{c} Different \\ rectangles \end{tabular} \\[1ex] 
 \hline
 4 & $3\,z_1+2\,z_2$ & $[e^{i\frac{\pi}{8}}\frac{c}{\sqrt{2}},\frac{s}{\sqrt{2}},\frac{ic}{\sqrt{2}},e^{-i\frac{3\pi}{8}}\frac{s}{\sqrt{2}}]^T$ & \begin{tabular}{c} $[s,\frac{c}{\sqrt{2}},0]^T$ \\ $[0,-\frac{s}{\sqrt{2}},\frac{1}{\sqrt{2}}]^T$ \end{tabular} & \begin{tabular}{c} Different \\ rectangles \end{tabular} \\[1ex] 
 \hline
 4 & $z_1+3\,z_2$ & $[\frac{1+s+ic}{2\sqrt{2}},\frac{c+is-i}{2\sqrt{2}},\frac{c+is+i}{2\sqrt{2}},\frac{1-s-ic}{2\sqrt{2}}]^T$ & \begin{tabular}{c} $[c,s^2,0]^T$ \\ $[0,-c^2,s]^T$ \end{tabular} & \begin{tabular}{c} Different \\ rectangles \end{tabular} \\[1ex] 
 \hline
\end{tabular}
\caption{Localizable tetrahedral $n=2$ qubit bases at level $k=4$. For each geometric class, we report one representative phase polynomial, the corresponding fiducial state and the Bloch vectors of the qubit reductions. In this table we defined $c:= \cos\frac{\pi}{8}$, $s:= \sin\frac{\pi}{8}$.} 
\label{table:1}
\end{table*}

At level \( k = 4 \), we encounter a landscape that was completely out of reach for general bases~\cite{Pauwels2025}. Each of the three canonical geometries from \( k = 2 \) reappears in a refined form, and new, irregular tetrahedra emerge. Interestingly, while for $k \leq 3$, the geometries on both qubits are always the same, for $k=4$, we also find examples in which the geometries on the two qubits are different (different squares or different tetrahedra). We summarize all solutions in Table~\ref{table:1}. 

Interestingly, at level $k=4$, we encounter (up to local unitaries) a representative of the family of bases generated by the fiducial state $\ket{\psi_{\mathrm{EJM}}^\theta}$ in Eq.~\eqref{eq:cyril-family}, which interpolates between the Bell state measurement ($\theta = \pi/2$) and the EJM ($\theta = 0$). Namely, we obtain fiducial states that are local-unitarily equivalent to $\ket{\psi_{\mathrm{EJM}}^{\theta=\pi/4}}$: see the first line in Table~\ref{table:1}. This suggests that more members of this family appear at higher levels of the hierarchy.

Indeed, the fiducial state $\ket{\psi_{\mathrm{EJM}}^\theta}$ can be written in the normal form of Eq.~\eqref{eq:GenFidPhase} with the diagonal phase gate
\begin{equation}
    D_{\vec\alpha} =\mathrm{diag}(1, e^{i\pi/2}, e^{-i\pi/2}, e^{i(\theta+\pi/2)})\,,
\end{equation}
(see Eq.~\eqref{eq:cyril-family-phases}). 
Considering a value of $\theta \in [0,\pi/2]$ such that $\theta/\pi$ is rational, we may write $\theta = 2\pi\ell/2^m$, for some precision value $m\ge2$ and some integer $\ell$. One then has $\theta + \pi/2 = 2\pi(\ell+2^{m-2})/2^m$, and the diagonal unitary $D_{\vec\alpha}$ above gets generated by the phase polynomial
\begin{align}
\hspace{-3mm} f_{m}(z_1, z_2) = -2^{m-2} z_1 + 2^{m-2} z_2 + (\ell+2^{m-2}) z_1 z_2 \ \notag \\
\mod 2^m.
\end{align}
Noting that the quadratic term $z_1 z_2$ dominates in Eq.~\eqref{eq:refined-level}, one finds that the localization level of this family is given by
\begin{align}
k & = (m - \nu_2(\ell+2^{m-2}) - 1) + 2 \notag \\
& = 1 - \nu_2\big(\tfrac{\theta}{2\pi} + \tfrac14\big), \label{eq:k_value_EJMtheta}
\end{align}
where we extended the 2-adic valuation to rational numbers via $\nu_2(r/s) = \nu_2(r) - \nu_2(s)$ and used the fact that $\nu_2\big(\frac{2^m}{2\pi}(\theta+\pi/2)\big) = m + \nu_2\big(\tfrac{\theta}{2\pi} + \tfrac14\big)$; see Fig.~\ref{fig:EJMfamily}.
This recovers our earlier results for the Bell basis ($m=2,\ell=1 \Rightarrow k = 2$), the EJM ($\ell=0 \Rightarrow k = 3$), and defines an infinite family of regular tetrahedra of radius $\frac{\sqrt{3}}{2} \cos\theta = \frac{\sqrt{3}}{2} \cos\!\left(\frac{2\pi \ell}{2^m} \right)$, each localizable at the level $k$ given in Eq.~\eqref{eq:k_value_EJMtheta}. This example nicely illustrates the observation made in \cite{Pauwels2025}, that at higher levels of the hierarchy, the different allowed degrees of entanglement becomes denser.

We conjecture that there similarly exist other phase-polynomial families generalizing both the twisted Bell bases (parameterized by the twisting angle) and the $E_2$ family.

\begin{figure}[ht!]
\centering
\includegraphics[width=0.96\columnwidth]{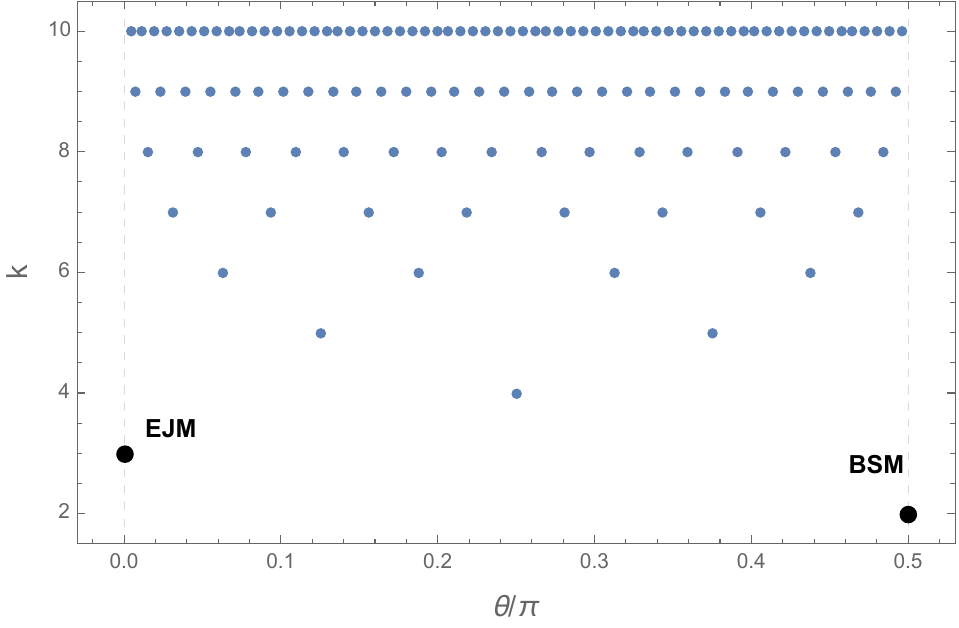}
\caption{Localization level $k$ of the family of orbit bases of $\ket{\psi^\theta_{\mathrm{EJM}}}$ (Eq.~\eqref{eq:cyril-family}) as a function of the interpolation angle $\theta$. The horizontal axis is shown in fractions of $\pi$. The Bell-state measurement (BSM, $\theta=\pi/2$) occurs at level $k=2$, while the EJM ($\theta=0$) occurs at level $k=3$. In between, rational values of $\theta/\pi$ yield discrete points with localization level $k = 1 - \nu_2\big(\tfrac{\theta}{2\pi} + \tfrac14\big)$, see Eq.~\eqref{eq:k_value_EJMtheta}.
Higher dyadic fractions thus interpolate between the two endpoints, producing an infinite family of localizable tetrahedral bases at increasing levels. \label{fig:EJMfamily}}
\end{figure}

\subsection{Three-qubit hierarchy}

We extend the systematic classification of tetrahedral bases to three-qubit systems. For three qubits, the general Boolean phase polynomial \( f_m : \mathbb{Z}_2^3 \to \mathbb{Z}_{2^m} \) has maximal algebraic degree 3, given explicitly by
\begin{align}
f_m(z_1,z_2,z_3) &=\; a_1 \,z_1 + a_2 \,z_2 + a_3 \,z_3 \nonumber \\
               &\quad + a_{12} \,z_1 \,z_2 + a_{13} \,z_1 \,z_3 + a_{23} \,z_2 \,z_3 \nonumber \\
               &\quad + a_{123} \,z_1 \,z_2 \,z_3 \mod 2^m.
\end{align}
The Clifford hierarchy level is again determined by Eq.~\eqref{eq:refined-level}.
We enumerate and classify all resulting tetrahedral bases at levels \(k = 1, 2, 3, 4\), organizing them according to their geometric structure. By the same arguments for the $n=2$ case, it is sufficient to consider precisions $m$ such that $k-2\le m\le k$. The full set of geometries is extensive and available in the computational appendix on GitHub~\cite{CompAppendix}. Here, we highlight several key features.

Like in the two-qubit case, at level $k=1$ ($m=1$) we again only find product bases. At level $k=2$ ($m=1,2$), we also obtain fiducial states with one qubit in a pure state, the other two in a maximally entangled state (e.g. of the form $\ket{\Phi^+}\ket{0}$) or with all three marginals completely mixed (e.g. of the form $\frac{1}{\sqrt{2}}(\ket{\Phi^+}\ket{0}+\ket{\Psi^+}\ket{1})$). At \(k = 3\), many bases resemble combinations or extensions of the known two-qubit geometries. For instance, the level-3 phase polynomial
\begin{equation} 
f_{m=1}(z_1, z_2, z_3) = z_1 z_2 z_3 \mod 2,
\end{equation} 
yields a configuration of squares, reminiscent of the \(E_2\) basis. Similarly, the polynomial
\begin{equation} 
f_{m=3}(z_1, z_2, z_3) = z_1 + z_2 + z_3 + 4z_1z_2z_3 \mod 2^3
\end{equation} 
produces a configuration with an irregular tetrahedron on the first qubit, a rectangle on the second and a square on the third.

At level \(k = 4\), we find several unique geometries, including a regular tetrahedral bases. This geometry has half the size of the two-qubit EJM. One example corresponds to the degree-3 polynomial at precision \(m = 2\),
\begin{equation} 
f_{m=2}(z_1, z_2, z_3) = z_1 z_3 + z_1 z_2 z_3 \mod 2^2 \,.
\end{equation} 
The associated fiducial state is explicitly given by
\begin{equation} 
\ket{\psi} = \frac{1}{2} \big[1,\ 0,\ 0,\ \tfrac{1 - i}{2},\ \tfrac{1 + i}{2},\ 1,\ 1,\ 0 \big]^T.
\end{equation} 

In Ref.~\cite{hm9n-mkb3}, we analyze in detail the solutions whose single-qubit marginals form regular tetrahedra on the Bloch sphere. For $n=3$, we find that although a geometrically unique regular configuration first appears at level $k=4$ of the hierarchy, it appears in multiple bases that are not related by local unitaries, reflecting the richness of multipartite entanglement. For $n=4$, the first such configurations—where all single-qubit Bloch vectors form regular tetrahedra—appear at level $k=5$.

\section{Discussion and Outlook}

Inspired by the Elegant Joint Measurement (EJM)~\cite{Gisin2019}—an example of a partially entangled measurement with symmetric local structure (tetrahedral single-qubit marginals) and interesting properties and applications—we developed a framework that abstracts and generalizes these features. Specifically, we constructed families of highly structured orthonormal bases as group orbits of fiducial states under tensor-product Weyl--Heisenberg actions, extending features of the EJM to multi-qubit and qudit systems. The resulting bases are locally encodable by construction and exhibit rich symmetry properties: geometric (e.g., tetrahedral or rectangular arrangements of Bloch vectors), algebraic (arising from group covariance), and entanglement-theoretic (isoentanglement). Moreover, we show that the problem of determining localization cost, which is generally hard~\cite{Pauwels2025}, admits a simple solution for this class of bases, enabling a general classification of the measurements we construct. Finally, our framework naturally recovers as special cases the families of two-qudit measurements studied in Refs.~\cite{Tavakoli2021,Czartowski2021}.

In contrast to the well-developed theory of entangled states, the structure of entangled measurements remains comparatively poorly understood, especially beyond the extremes of local and fully global maximally entangled measurements~\cite{Cavalcanti2023}. Explicit, symmetric “intermediate'' examples are scarce, despite their practical and conceptual importance~\cite{Gisin2019,Tavakoli2022}. In this context, structured families of joint measurements---exhibiting symmetries and admitting a transparent description---are valuable both for theory (they enable tractable analysis) and for practice (they suggest implementable targets).
From this perspective, one immediate application of the present work is methodological: for the tetrahedral families introduced here, questions of (efficient) localizability and entanglement cost become analyzable in a way that is typically out of reach for generic measurements. Indeed, outside highly structured settings, determining whether a given multipartite projective measurement can be implemented under restricted communication models is essentially intractable; here, symmetry and phase-gate parametrization reduce the problem to a controlled analysis in terms of the underlying diagonal structure and its position in the relevant hierarchy.

A second motivation comes from quantum networks. Joint measurements are the operational ``connectors'' that turn independently distributed resources into nontrivial network correlations, and their symmetries often translate directly into symmetries of the induced correlations, which can dramatically simplify both analytic and numerical studies~\cite{Gitton2024}. This was one of the original motivations for introducing the EJM as a highly symmetric intermediate measurement~\cite{Gisin2019}, and related work has already explored how tuning the entanglement within a structured EJM-like family affects robustness of network nonlocality in bilocal scenarios~\cite{Tavakoli2021}. The present three-parameter tetrahedral family in Eq.~\eqref{eq:bell-superposition} provides a richer testbed for extending such questions beyond one-dimensional interpolations (e.g., beyond the BSM--EJM path), and for systematically exploring how symmetry, entanglement, and implementational complexity trade off in network tasks. Finally, in ``nonlocality-without-inputs'' settings, genuinely non-Clifford elements are required to generate nonclassical correlations~\cite{GattoLamas2023}; our constructions supply explicit low-complexity non-Clifford candidates with strong symmetry constraints, making them natural targets for both theory and experiment.

Aside from investigating these applications, several open directions emerge from this study. First, a natural goal is to develop a systematic method for imposing additional local geometric constraints on the single-qubit reductions of fiducial states. In particular, identifying general conditions under which the single-qubit reductions lie on specific geometric configurations—such as regular tetrahedra, squares, or other symmetric constellations. 

A second avenue concerns the classification of localizable tetrahedral bases in higher dimensions. While the two-qudit measurement proposed in Ref.~\cite{Czartowski2021} generalizes some geometric features of EJM, it is not efficiently localizable. Can we find a family of two-qudit bases that is? Establishing such constructions would be of practical interest for entanglement-assisted protocols in quantum networks and for implementing symmetric measurements in measurement-based quantum computation.

One may ask whether the local symmetry exploited here—tetrahedral symmetry generated by the single-qubit Pauli group—admits meaningful generalizations. While the Pauli group preserves the tetrahedron, it does not generate all polyhedral symmetries—for example, mapping the tetrahedron to its dual requires rotations outside the group. The Clifford group, for example, includes such transformations, mapping Pauli axes onto each other and sending the tetrahedron to the octahedron. With suitable fiducial states, Clifford-generated orbits could therefore produce bases whose Bloch vectors lie on other regular solids. However, our construction relies essentially on the Pauli group being projectively abelian—tensor products commute up to a phase, allowing these phases to cancel and Schur's lemma to enforce basis completeness. Clifford tensor-product representations are non-abelian and lack this structure. More generally, among finite subgroups of $SU(2)$, only cyclic groups and the Pauli group have this property, suggesting that the Pauli group is uniquely suited for generating locally covariant, symmetric bases via group orbits.

While our construction of party-permutation-invariant (PPI) bases focused on odd \( n \), it remains open whether similar symmetric structures exist for even numbers of qubits. The obstruction stems from the specific abelian Pauli subgroup used, which fails to support PPI fiducial states in the even case. However, alternative abelian subgroups of order \( 2^n \) may still allow the construction of party-permutation-invariant orthonormal bases for even \( n \). 

Finally, the construction presented here can be viewed as a particular instance of a more general framework for building measurement bases based on a particular choice of stabilizer group. It is an interesting open question to explore how the choice of stabilizer group influences the entanglement and geometric properties of the resulting measurement bases. We leave this exploration for future work.

We conclude that the interplay between symmetry, local encodability, and localization complexity offers a toolbox for constructing and classifying measurement bases. The use of projectively abelian groups, discrete geometric symmetries, and Clifford hierarchy constraints opens the door to new classes of measurements with practical and theoretical relevance. Future work may reveal deeper connections between these constructions and broader topics such as SIC-POVMs and measurement-based quantum computing.

\textbf{Acknowledgements.}
We thank Matthias Christandl, Otfried Gühne, Teiko Heinosaari, Pauli Jokinen and Pavel Sekatski for discussions.
We acknowledge financial support from Swiss National Science Foundation (NCCR-SwissMAP and grant number 224561).

\bibliography{refs.bib}

@Article{Popp2024,
  author    = {Popp, Christopher and Hiesmayr, Beatrix C},
  journal   = {New J. Phys.},
  title     = {Special features of the {W}eyl–{H}eisenberg {B}ell basis imply unusual entanglement structure of {B}ell-diagonal states},
  year      = {2024},
  issn      = {1367-2630},
  month     = jan,
  number    = {1},
  pages     = {013039},
  volume    = {26},
  doi       = {10.1088/1367-2630/ad1d0e},
  publisher = {IOP Publishing},
}

@book{Waldron2018,
  title = {An Introduction to Finite Tight Frames},
  ISBN = {9780817648152},
  ISSN = {2296-5017},
  url = {http://dx.doi.org/10.1007/978-0-8176-4815-2},
  DOI = {10.1007/978-0-8176-4815-2},
  journal = {Applied and Numerical Harmonic Analysis},
  publisher = {Springer New York},
  author = {Waldron,  Shayne F. D.},
  year = {2018}
}

@article{GattoLamas2023,
   title={Multipartite Nonlocality in Clifford Networks},
   volume={130},
   ISSN={1079-7114},
   url={http://dx.doi.org/10.1103/PhysRevLett.130.240802},
   DOI={10.1103/physrevlett.130.240802},
   number={24},
   journal={Physical Review Letters},
   publisher={American Physical Society (APS)},
   author={Gatto Lamas, Amanda and Chitambar, Eric},
   year={2023},
   month=jun }

@article{Toth2005,
  title = {Entanglement detection in the stabilizer formalism},
  author = {T\'oth, G\'eza and G\"uhne, Otfried},
  journal = {Phys. Rev. A},
  volume = {72},
  issue = {2},
  pages = {022340},
  numpages = {14},
  year = {2005},
  month = {Aug},
  publisher = {American Physical Society},
  doi = {10.1103/PhysRevA.72.022340},
  url = {https://link.aps.org/doi/10.1103/PhysRevA.72.022340}
}

@InBook{Krasikov2001,
  author    = {Krasikov, Ilia and Litsyn, Simon},
  pages     = {199--211},
  publisher = {American Mathematical Society},
  title     = {Survey of binary Krawtchouk polynomials},
  year      = {2001},
  month     = feb,
  booktitle = {Codes and Association Schemes},
  doi       = {10.1090/dimacs/056/16},
  issn      = {2472-4793},
}

@article{Lemmens1973,
  title = {Equiangular lines},
  volume = {24},
  ISSN = {0021-8693},
  url = {http://dx.doi.org/10.1016/0021-8693(73)90123-3},
  DOI = {10.1016/0021-8693(73)90123-3},
  number = {3},
  journal = {J. Algebra},
  publisher = {Elsevier BV},
  author = {Lemmens,  P.W.H and Seidel,  J.J},
  year = {1973},
  month = mar,
  pages = {494–512}
}

@phdthesis{Zauner1999,
  author    = {Zauner, G.},
  title     = {Quantendesigns: Grundzüge einer nichtkommutativen Designtheorie},
  school    = {University of Vienna},
  year      = {1999},
  note      = {English translation available at \href{https://doi.org/10.1142/S0219749911006776}{Int. J. Quantum Inf. {\bf 09}, 445--507 (2011)} }
}

@article{Renes2004,
   title={Symmetric informationally complete quantum measurements},
   volume={45},
   ISSN={1089-7658},
   url={http://dx.doi.org/10.1063/1.1737053},
   DOI={10.1063/1.1737053},
   number={6},
   journal={J. Math. Phys.},
   publisher={AIP Publishing},
   author={Renes, Joseph M. and Blume-Kohout, Robin and Scott, A. J. and Caves, Carlton M.},
   year={2004},
   month=jun, pages={2171–2180} }

@article{Tanaka007,
   title={Local encoding of classical information onto quantum states},
   volume={54},
   ISSN={1362-3044},
   url={http://dx.doi.org/10.1080/09500340701403301},
   DOI={10.1080/09500340701403301},
   number={13–15},
   journal={J. Mod. Opt.},
   publisher={Informa UK Limited},
   author={Tanaka, Yu and Markham, Damian and Murao, Mio},
   year={2007},
   month=sep, pages={2259–2273} }

@phdthesis{Chiribella2004,
  author    = {Chiribella, Giulio},
  title     = {Optimal Estimation of Quantum Signals in the Presence of Symmetry},
  school    = {Università degli Studi di Pavia},
  year      = {2004},
  note      = {Available at \url{https://www.quantumtechnologies.it/wp-content/uploads/thesis/ThesisRevised.pdf}}
}

@article{Tavakoli2021,
  title = {Bilocal {B}ell Inequalities Violated by the Quantum Elegant Joint Measurement},
  author = {Tavakoli, Armin and Gisin, Nicolas and Branciard, Cyril},
  journal = {Phys. Rev. Lett.},
  volume = {126},
  issue = {22},
  pages = {220401},
  numpages = {6},
  year = {2021},
  month = {Jun},
  publisher = {American Physical Society},
  doi = {10.1103/PhysRevLett.126.220401},
  url = {https://link.aps.org/doi/10.1103/PhysRevLett.126.220401}
}

@Article{Gisin2019,
  author    = {Gisin, Nicolas},
  journal   = {Entropy},
  title     = {Entanglement 25 Years after Quantum Teleportation: Testing Joint Measurements in Quantum Networks},
  year      = {2019},
  issn      = {1099-4300},
  month     = mar,
  number    = {3},
  pages     = {325},
  volume    = {21},
  doi       = {10.3390/e21030325},
  publisher = {MDPI AG},
}

@Article{Czartowski2021,
  author    = {Czartowski, Jakub and Życzkowski, Karol},
  journal   = {Quantum},
  title     = {Bipartite quantum measurements with optimal single-sided distinguishability},
  year      = {2021},
  issn      = {2521-327X},
  month     = apr,
  pages     = {442},
  volume    = {5},
  doi       = {10.22331/q-2021-04-26-442},
  publisher = {Verein zur Forderung des Open Access Publizierens in den Quantenwissenschaften},
}

@article{Pauwels2025,
  title = {Classification of Joint Quantum Measurements Based on Entanglement Cost of Localization},
  author = {Pauwels, Jef and Pozas-Kerstjens, Alejandro and {del} Santo, Flavio and Gisin, Nicolas},
  journal = {Phys. Rev. X},
  volume = {15},
  issue = {2},
  pages = {021013},
  numpages = {20},
  year = {2025},
  month = {Apr},
  publisher = {American Physical Society},
  doi = {10.1103/PhysRevX.15.021013},
  url = {https://link.aps.org/doi/10.1103/PhysRevX.15.021013}
}

@Article{Vaidman2003,
  author    = {Vaidman, Lev},
  journal   = {Phys. Rev. Lett.},
  title     = {Instantaneous Measurement of Nonlocal Variables},
  year      = {2003},
  issn      = {1079-7114},
  month     = jan,
  number    = {1},
  pages     = {010402},
  volume    = {90},
  doi       = {10.1103/physrevlett.90.010402},
  publisher = {American Physical Society (APS)},
}

@Article{Groisman2003,
  author    = {Groisman, Berry and Reznik, Benni and Vaidman, Lev},
  journal   = {J. Mod. Opt.},
  title     = {Instantaneous measurements of nonlocal variables},
  year      = {2003},
  issn      = {1362-3044},
  month     = apr,
  number    = {6–7},
  pages     = {943--949},
  volume    = {50},
  doi       = {10.1080/09500340308234543},
  publisher = {Informa UK Limited},
}

@Article{Tavakoli2021a,
  author        = {Tavakoli, Armin and Pauwels, Jef and Woodhead, Erik and Pironio, Stefano},
  journal       = {PRX Quantum},
  title         = {Correlations in Entanglement-Assisted Prepare-and-Measure Scenarios},
  year          = {2021},
  month         = {Dec},
  pages         = {040357},
  volume        = {2},
  doi           = {10.1103/PRXQuantum.2.040357},
  issue         = {4},
  numpages      = {17},
  publisher     = {American Physical Society}
}

@Book{Bell1988,
  author    = {Bell, John S.},
  publisher = {Cambridge University Press},
  title     = {Speakable and Unspeakable in Quantum Mechanics: Collected Papers on Quantum Philosophy},
  year      = {1988},
  url       = {https://doi.org/10.1017/CBO9780511815676},
}

@article{Horodecki2009,
  title = {Quantum entanglement},
  author = {Horodecki, Ryszard and Horodecki, Pawe\l{} and Horodecki, Micha\l{} and Horodecki, Karol},
  journal = {Rev. Mod. Phys.},
  volume = {81},
  issue = {2},
  pages = {865--942},
  numpages = {0},
  year = {2009},
  month = {Jun},
  publisher = {American Physical Society},
  doi = {10.1103/RevModPhys.81.865},
  url = {https://link.aps.org/doi/10.1103/RevModPhys.81.865}
}

@article{Jozsa2003,
   title={On the role of entanglement in quantum-computational speed-up},
   volume={459},
   ISSN={1471-2946},
   url={http://dx.doi.org/10.1098/rspa.2002.1097},
   DOI={10.1098/rspa.2002.1097},
   number={2036},
   journal={Proc. R. Soc. Lond. A.},
   publisher={The Royal Society},
   author={Jozsa, Richard and Linden, Noah},
   year={2003},
   month=aug, pages={2011–2032} }

@article{Vazirani2014,
  title = {Fully Device-Independent Quantum Key Distribution},
  author = {Vazirani, Umesh and Vidick, Thomas},
  journal = {Phys. Rev. Lett.},
  volume = {113},
  issue = {14},
  pages = {140501},
  numpages = {6},
  year = {2014},
  month = {Sep},
  publisher = {American Physical Society},
  doi = {10.1103/PhysRevLett.113.140501},
  url = {https://link.aps.org/doi/10.1103/PhysRevLett.113.140501}
}

@Article{Cavalcanti2023,
  author    = {Cavalcanti, Eric G. and Chaves, Rafael and Giacomini, Flaminia and Liang, Yeong-Cherng},
  journal   = {Nat. Rev. Phys.},
  title     = {Fresh perspectives on the foundations of quantum physics},
  year      = {2023},
  issn      = {2522-5820},
  month     = apr,
  number    = {6},
  pages     = {323--325},
  volume    = {5},
  doi       = {10.1038/s42254-023-00586-z},
  publisher = {Springer Science and Business Media LLC},
}

@article{Tavakoli2022,
   title={Bell nonlocality in networks},
   volume={85},
   ISSN={1361-6633},
   url={http://dx.doi.org/10.1088/1361-6633/ac41bb},
   DOI={10.1088/1361-6633/ac41bb},
   number={5},
   journal={Rep. Prog. Phys.},
   publisher={IOP Publishing},
   author={Tavakoli, Armin and Pozas-Kerstjens, Alejandro and Luo, Ming-Xing and Renou, Marc-Olivier},
   year={2022},
   month=mar, pages={056001} }

@article{Caleffi2024,
   title={Distributed quantum computing: A survey},
   volume={254},
   ISSN={1389-1286},
   url={http://dx.doi.org/10.1016/j.comnet.2024.110672},
   DOI={10.1016/j.comnet.2024.110672},
   journal={Comput. Netw.},
   publisher={Elsevier BV},
   author={Caleffi, Marcello and Amoretti, Michele and Ferrari, Davide and Illiano, Jessica and Manzalini, Antonio and Cacciapuoti, Angela Sara},
   year={2024},
   month=dec, pages={110672} }

@article{Briegel1998,
  title = {Quantum Repeaters: The Role of Imperfect Local Operations in Quantum Communication},
  author = {Briegel, Hans J. and D\"ur, Wolfgang and Cirac, J. Ignacio and Zoller, Peter},
  journal = {Phys. Rev. Lett.},
  volume = {81},
  issue = {26},
  pages = {5932--5935},
  numpages = {0},
  year = {1998},
  month = {Dec},
  publisher = {American Physical Society},
  doi = {10.1103/PhysRevLett.81.5932},
  url = {https://link.aps.org/doi/10.1103/PhysRevLett.81.5932}
}

@unpublished{Gitton2024,
       author = {{Gitton}, Victor and {Renner}, Renato},
        title = "{The Elegant Joint Measurement is Non-Classical in the Triangle Network}",
         year = 2025,
        month = oct,
archivePrefix = {arXiv},
       eprint = {2510.15143},
 primaryClass = {quant-ph},
}

@article{Schur1905,
  author    = {Issai Schur},
  title     = {Neue Begrundung der Theorie der Gruppencharaktere},
  journal   = {Sitzungsberichte der Koenigl. Preuss. Akad. Wiss. Berlin},
  year      = {1905},
  note      = {{E}nglish title: New foundation for the theory of group characters}
}

@article{Gross2006,
  author    = {David Gross},
  title     = {Hudson's theorem for finite-dimensional quantum systems},
  journal   = {J. Math. Phys.},
  volume    = {47},
  number    = {12},
  pages     = {122107},
  year      = {2006},
  doi       = {10.1063/1.2393152}
}

@ARTICLE{Chiribella2006,
       author = {{Chiribella}, Giulio and {Mauro D'Ariano}, Giacomo},
        title = "{Extremal covariant measurements}",
      journal = {J. Math. Phys.},
         year = 2006,
        month = sep,
       volume = {47},
       number = {9},
        pages = {092107-092107},
          doi = {10.1063/1.2349481},
       adsurl = {https://ui.adsabs.harvard.edu/abs/2006JMP....47i2107C}
}

@article{Chiribella2004a,
  title = {Covariant quantum measurements that maximize the likelihood},
  author = {Chiribella, Giulio and D'Ariano, Giacomo Mauro and Perinotti, Paolo and Sacchi, Massimiliano F.},
  journal = {Phys. Rev. A},
  volume = {70},
  issue = {6},
  pages = {062105},
  numpages = {8},
  year = {2004},
  month = {Dec},
  publisher = {American Physical Society},
  doi = {10.1103/PhysRevA.70.062105},}

@misc{Holevo1982,
  author       = {A. S. Holevo},
  title        = {Probabilistic and Statistical Aspects of Quantum Theory},
  publisher    = {North-Holland Publishing Co.},
  address      = {Amsterdam},
  year         = {1982},
  isbn         = {9780444863339},
  note         = {Translated from the Russian by R. A. Silverman. Originally published in Russian in 1974},
}

@Article{Gottesman1999,
  author    = {Gottesman, Daniel and Chuang, Isaac L.},
  journal   = {Nature},
  title     = {Demonstrating the viability of universal quantum computation using teleportation and single-qubit operations},
  year      = {1999},
  issn      = {1476-4687},
  month     = nov,
  number    = {6760},
  pages     = {390--393},
  volume    = {402},
  doi       = {10.1038/46503},
  publisher = {Springer Science and Business Media LLC},
}

@phdthesis{Gottesman1997,
  author       = {Daniel E. Gottesman},
  title        = {Stabilizer Codes and Quantum Error Correction},
  school       = {California Institute of Technology},
  year         = {1997},
  doi          = {10.7907/RZR7-DT72}
}

@unpublished{silva2025,
    title={The Clifford hierarchy for one qubit or qudit},
    author={Nadish de Silva and Oscar Lautsch},
    year={2025},
    eprint={2501.07939},
    archivePrefix={arXiv},
    primaryClass={quant-ph}
}

@article{Gottesman2017,
  title = {Diagonal gates in the Clifford hierarchy},
  author = {Cui, Shawn X. and Gottesman, Daniel and Krishna, Anirudh},
  journal = {Phys. Rev. A},
  volume = {95},
  issue = {1},
  pages = {012329},
  numpages = {7},
  year = {2017},
  month = {Jan},
  publisher = {American Physical Society},
  doi = {10.1103/PhysRevA.95.012329},
  url = {https://link.aps.org/doi/10.1103/PhysRevA.95.012329}
}

@article{Bennet1993,
  title = {Teleporting an unknown quantum state via dual classical and Einstein-Podolsky-Rosen channels},
  author = {Bennett, Charles H. and Brassard, Gilles and Cr\'epeau, Claude and Jozsa, Richard and Peres, Asher and Wootters, William K.},
  journal = {Phys. Rev. Lett.},
  volume = {70},
  issue = {13},
  pages = {1895--1899},
  numpages = {0},
  year = {1993},
  month = {Mar},
  publisher = {American Physical Society},
  doi = {10.1103/PhysRevLett.70.1895},
  url = {https://link.aps.org/doi/10.1103/PhysRevLett.70.1895}
}

@Article{Piveteau2022,
  author        = {{Piveteau}, Am{\'e}lie and {Pauwels}, Jef and {H{\^a}kansson}, Emil and {Muhammad}, Sadiq and {Bourennane}, Mohamed and {Tavakoli}, Armin},
  journal       = {Nat. Commun.},
  title         = {{Entanglement-assisted quantum communication with simple measurements}},
  year          = {2022},
  month         = dec,
  pages         = {7878},
  volume        = {13},
  adsurl        = {https://ui.adsabs.harvard.edu/abs/2022NatCo..13.7878P},
  archiveprefix = {arXiv},
  doi           = {10.1038/s41467-022-33922-5},
  eid           = {7878}
}

@Article{DelSanto2024,
  author    = {Del Santo, Flavio and Czartowski, Jakub and \ifmmode \dot{Z}\else \.{Z}\fi{}yczkowski, Karol and Gisin, Nicolas},
  journal   = {Phys. Rev. Res.},
  title     = {Iso-entangled bases and joint measurements},
  year      = {2024},
  month     = {Apr},
  pages     = {023085},
  volume    = {6},
  doi       = {10.1103/PhysRevResearch.6.023085},
  issue     = {2},
  numpages  = {9},
  publisher = {American Physical Society},
  url       = {https://link.aps.org/doi/10.1103/PhysRevResearch.6.023085},
}

@Article{Popescu1994,
  author    = {Popescu, Sandu and Vaidman, Lev},
  journal   = {Phys. Rev. A},
  title     = {Causality constraints on nonlocal quantum measurements},
  year      = {1994},
  issn      = {1094-1622},
  month     = jun,
  number    = {6},
  pages     = {4331--4338},
  volume    = {49},
  doi       = {10.1103/physreva.49.4331},
  publisher = {American Physical Society (APS)},
}

@article{Pimpel2023,
  title = {Correspondence between entangled states and entangled bases under local transformations},
  author = {Pimpel, Florian and Renner, Martin J. and Tavakoli, Armin},
  journal = {Phys. Rev. A},
  volume = {108},
  issue = {2},
  pages = {022220},
  numpages = {9},
  year = {2023},
  month = {Aug},
  publisher = {American Physical Society},
  doi = {10.1103/PhysRevA.108.022220},
  url = {https://link.aps.org/doi/10.1103/PhysRevA.108.022220}
}

@unpublished{zhang2025,
    title={Quantum stochastic communication via high-dimensional entanglement},
    author={Chao Zhang and Jia-Le Miao and Xiao-Min Hu and Jef Pauwels and Yu Guo and Chuan-Feng Li and Guang-Can Guo and Armin Tavakoli and Bi-Heng Liu},
    year={2025},
    eprint={2502.04887},
    archivePrefix={arXiv},
    primaryClass={quant-ph}
}

@unpublished{Boreiri2023,
  author    = {Boreiri, Sadra and Ulu, Bora and Brunner, Nicolas and Sekatski, Pavel},
  title     = {Noise-robust proofs of quantum network nonlocality},
  year      = {2023},
  copyright = {Creative Commons Attribution 4.0 International},
  doi       = {10.48550/ARXIV.2311.02182},
  publisher = {arXiv},
    eprint={2311.02182},
    archivePrefix={arXiv},
    primaryClass={quant-ph}
}

@book{coxeter1973,
  title={Regular Polytopes},
  author={Coxeter, H.S.M.},
  isbn={9780486614809},
  lccn={73084364},
  series={Dover books on advanced mathematics},
  url={https://books.google.ch/books?id=iWvXsVInpgMC},
  year={1973},
  publisher={Dover Publications}
}

@article{Kraus2009,
  title = {Local entanglability and multipartite entanglement},
  author = {Kruszynska, C. and Kraus, B.},
  journal = {Phys. Rev. A},
  volume = {79},
  issue = {5},
  pages = {052304},
  numpages = {4},
  year = {2009},
  month = {May},
  publisher = {American Physical Society},
  doi = {10.1103/PhysRevA.79.052304},
  url = {https://link.aps.org/doi/10.1103/PhysRevA.79.052304}
}

@ARTICLE{Dang2013,
       author = {{Dang}, Hoan Bui and {Blanchfield}, Kate and {Bengtsson}, Ingemar and {Appleby}, D.~M.},
        title = "{Linear dependencies in Weyl-Heisenberg orbits}",
      journal = {Quantum Inf. Process.},
         year = 2013,
        month = nov,
       volume = {12},
       number = {11},
        pages = {3449-3475},
          doi = {10.1007/s11128-013-0609-6}
}

@article{hm9n-mkb3,
  title = {Multiqubit Elegant Joint Measurement},
  author = {Pauwels, Jef and Gisin, Nicolas},
  journal = {Phys. Rev. Lett.},
  volume = {136},
  issue = {19},
  pages = {190201},
  numpages = {7},
  year = {2026},
  month = {May},
  publisher = {American Physical Society},
  doi = {10.1103/hm9n-mkb3},
  url = {https://link.aps.org/doi/10.1103/hm9n-mkb3}
}

@misc{CompAppendix,
  author       = {Jef Pauwels},
  title        = {Computational Appendix},
  year         = {2025},
  howpublished = {\url{https://github.com/jefpauwels/TetrahedralBases}},
  note         = {Version: 2025-07-22}
}

\clearpage

\appendix

\section{Joint eigenbases for the elements of $G_{\mathrm{tetra}}^{(n)}$ and $G_d^{(n)}$}
\label{app:eigenbases}

In this appendix, we show explicit how one may compute the joint eigenbases of $G_d^{(n)}$ ($=G_{\mathrm{tetra}}^{(n)}$ for $d=2$).

For simplicity, we start with the $n$-qubit case. Rather than diagonalizing \( G_{\mathrm{tetra}}^{(n)} \) directly, we proceed in two steps. First, recalling that $G_{\mathrm{tetra}}^{(2)}\cong\mathbb{Z}_2^{\,n}$, we identify a simpler representation of the same group 
\begin{equation}
  G_{\vec Z X}^{(n)}
  \;=\;
  \bigl\langle
      Z^{(1)},\,Z^{(2)},\,\dots,\,Z^{(n-1)},\,X^{(n)}
  \bigr\rangle
  \;\cong\;
  \mathbb{Z}_2^{\,n}.
\end{equation}
Its simultaneous eigenbasis is the set of product states:
\begin{equation}
  \ket{\vec{z},x}_{\vec Z X}
  \;=\;
  \bigotimes_{i=1}^{n-1} \ket{z_i}_Z \otimes \ket{x}_X,
\end{equation}
where $\vec{z}=(z_1,\ldots,z_{n-1}) \in \{0,1\}^{n-1}, x \in \{0,1\}$ and with $\ket{\cdot}_Z$ and $\ket{\cdot}_X$ denoting eigenstates of the $Z$ and $X$ Pauli operators, respectively (i.e., $\ket{0/1}_Z$ are the computational basis states $\ket{0/1}$, and $\ket{x}_X = \frac{1}{\sqrt{2}}(\ket{0}+(-1)^x\ket{1})$).

Next, introduce the “right-to-left” CNOT chain
\begin{equation}
S(n) = \mathrm{CNOT}_{(2 \to 1)}\cdots\mathrm{CNOT}_{(n \to n-1)},
\end{equation}
as in Eq.~\eqref{eq:CNOTchain} of the main text.
One can verify that
\begin{align}
  S(n)\,Z^{(i)}\,S(n)^\dagger &= Z^{(i)} Z^{(i+1)} \qquad \text{for } i < n, \nonumber \\
  S(n)\,X^{(n)}\,S(n)^\dagger &= X^{\otimes n}.
\end{align}
This follows from standard CNOT conjugation rules: \( \mathrm{CNOT}_{(i \to j)} \) maps \( Z^{(j)} \mapsto Z^{(i)} Z^{(j)} \) and \( X^{(i)} \mapsto X^{(i)} X^{(j)} \), 
thereby propagating \( Z \) forward and \( X \) backward along the chain.
Hence \( S(n) \) conjugates \( G_{\vec Z X}^{(n)} \) to \( G_{\mathrm{tetra}}^{(n)} \).

Applying \( S(n) \) to each product eigenstate then yields the joint eigenbasis of \( G_{\mathrm{tetra}}^{(n)} \), as in Eq.~\eqref{eq:FourierTetra}:
\begin{equation}
    \ket{\Phi_{\vec{z},x}}
    = S(n)\;
    \ket{\vec{z},x}_{\vec Z X}, \quad (\vec{z},x) \in \{0,1\}^n.
\end{equation}
More explicitly, one finds
\begin{align}
    \ket{\Phi_{\vec{z},x}}
    = \tfrac{1}{\sqrt{2}} & \big( \ket{s_1,\ldots,s_{n-1},0} \notag \\
    & +(-1)^x \ket{s_1\oplus1,\ldots,s_{n-1}\oplus1,1} \!\big),
\end{align}
with each $s_i:=\bigoplus_{j=i}^{n-1}z_j$.
Notice that each of these states is entangled (of the GHZ type) for any \( n \geq 2 \), and satisfies
\begin{align}
Z^{(i)} Z^{(i+1)}\ket{\Phi_{\vec{z},x}} &= (-1)^{z_i} \ket{\Phi_{\vec{z},x}}, \notag \\
X^{\otimes n} \ket{\Phi_{\vec{z},x}} &= (-1)^{x} \ket{\Phi_{\vec{z},x}}.
\end{align}

\medskip

The qu$d$it case follows precisely the same logic. That is, we start from  the simpler representation \(\langle Z_d^{(1)}, \dots, Z_d^{(n-1)}, X_d^{(n)} \rangle\) using the following sequence of SUM gates (as defined in Eq.~\eqref{eq:SUMchain} of the main text):
\begin{equation}
S_d(n) = \mathrm{SUM}_{(2 \to 1)}\, \mathrm{SUM}_{(3 \to 2)}\, \cdots\, \mathrm{SUM}_{(n \to n-1)}.
\end{equation}
These satisfy
\begin{align} \label{eq:change_group_rep}
  S_d(n)\,Z_d^{(i)}\,S_d(n)^\dagger &= Z_d^{(i)} {Z_d^{(i+1)}}^* \qquad \text{for } i < n, \nonumber \\
  S_d(n) X_d^{(n)} S_d(n)^\dagger & = S_d(n) {H_d^{(n)}}^\dagger Z_d^{(n)} H_d^{(n)} S_d(n)^\dagger = X_d^{\otimes n},
\end{align}
where $H_d =\sum_{j,k=0}^{d-1} \omega_d^{j k}\ketbra{j}{k}$ (written in the computational basis, with $\omega_d = e^{2i\pi/d}$) is the generalized $d$-dimensional Hadamard gate, or discrete Fourier transform.
Applying \(S_d(n)\) to the product basis
\begin{equation}
\ket{\vec{z},x}_{\vec Z X} = 
\bigotimes_{i=1}^{n-1} \ket{z_i}_{Z_d} \otimes \ket{x}_{X_d},
\quad
\vec{z} \in \mathbb{Z}_d^{n-1},\;
x \in \mathbb{Z}_d,
\end{equation}
with $\ket{\cdot}_{Z_d}$ and $\ket{\cdot}_{X_d}$ denoting eigenstates of the $Z_d$ and $X_d$ Pauli operators (i.e., the former corresponding to the computational basis and the latter being defined as $\ket{x}_{X_d} = H_d^\dagger \ket{x}_{Z_d} = \frac{1}{\sqrt{d}}\sum_{j=0}^{d-1} \omega_d^{-jx}\ket{j}$, such that $X_d \ket{x}_{X_d} = \omega_d^x \ket{x}_{X_d}$)
yields the joint eigenbasis of \(G_d^{(n)}\), as in Eq.~\eqref{eq:Fourier-qudit-final}:
\begin{equation}
\ket{\Phi^{(d)}_{\vec{z},x}} = S_d(n) \ket{\vec{z},x}_{\vec Z X},
\quad
(\vec{z},x) \in \mathbb{Z}_d^n.
\end{equation}
More explicitly, one finds
\begin{align}
    \ket{\Phi^{(d)}_{\vec{z},x}}
    = \frac{1}{\sqrt{d}} \sum_{k=0}^{d-1} \omega_d^{-k x} \ket{s_1\oplus k,\ldots,s_{n-1}\oplus k,k},
\end{align}
with again each $s_i:=\bigoplus_{j=i}^{n-1}z_j$ (and with $\oplus$ denoting here addition modulo $d$).
These satisfy
\begin{align}
Z_d^{(i)} {Z_d^{(i+1)}}^*\ket{\Phi^{(d)}_{\vec{z},x}} &= \omega_d^{z_i} \ket{\Phi^{(d)}_{\vec{z},x}}, \notag \\
X_d^{\otimes n} \ket{\Phi^{(d)}_{\vec{z},x}} &= \omega_d^{x} \ket{\Phi^{(d)}_{\vec{z},x}}.
\end{align}

\section{Party-Permutation-Invariant fiducial states}
\label{app:PPI}

In this appendix we justify the form of $\ket{\psi_{\mathrm{PPI}}^{(n)}}$, with its parameters $a_k$ and $\alpha_k$ satisfying Eq.~\eqref{eq:psiPPI_amplitudes_phases} for odd $n$, and show that no PPI fiducial state exists for even $n$.

Consider the super-normalized Dicke states $\ket{D_k} = \sum_{\vec z:\mathrm{wt}(\vec z)=k} \ket{\vec z}$. One can first show that
\begin{align}
    & \bra{D_k}(\underbrace{Z\cdots Z}_{m}\underbrace{\mathds{1}\cdots\mathds{1}}_{n-m})\ket{D_{k'}} = \delta_{k,k'} \,K^n_k(m), \notag \\
    & \bra{D_k}(\underbrace{Z\cdots Z}_{m}\underbrace{\mathds{1}\cdots\mathds{1}}_{n-m})(\underbrace{X\cdots X}_n)\ket{D_{k'}} = \delta_{k,n-k'} \,K^n_k(m), \label{eq:DkZZIIDk}
\end{align}
where we introduced the so-called binary Krawtchouk polynomials~\cite{Krasikov2001}
\begin{equation} \label{eq:def_Krawtchouk}
    K^n_k(m) := \sum_{m_1=0}^m (-1)^{m_1} \binom{m}{m_1} \binom{n-m}{k-m_1}. 
\end{equation}
The first line of Eq.~\eqref{eq:DkZZIIDk} can be obtained by noting that
\(\bra{D_k}(Z\cdots Z\,\mathds{1}\cdots\mathds{1})\ket{D_{k'}} 
=\delta_{k,k'} \,\sum_{\vec z:\,\mathrm{wt}(\vec z)=k} \bra{\vec z}(Z\cdots Z\,\mathds{1}\cdots\mathds{1})\ket{\vec z}\), in which each term \(\bra{\vec z}(Z\cdots Z\,\mathds{1}\cdots\mathds{1})\ket{\vec z}\) is equal to \((-1)^{m_1}\), 
where \(m_1\) is the number of \(1\)’s in the first \(m\) positions of \(\vec z\). 
The number of such \(\vec z\) with Hamming weight \(k\) and exactly \(m_1\) \(1\)’s in the first \(m\) positions is 
\(\binom{m}{m_1} \binom{n-m}{\,k-m_1\,}\). 
Summing over \(m_1\) yields the expression for \(K^n_k(m)\) in Eq.~\eqref{eq:def_Krawtchouk}. 
The second line of Eq.~\eqref{eq:DkZZIIDk} then follows from the first by noting that 
\((X\cdots X)\ket{D_{k'}} = \ket{D_{n-k'}}\).

For a PPI state
\begin{equation}
\ket{\psi_{\mathrm{PPI}}^{(n)}} = \sum_{k=0}^{n} a_k\, e^{i \alpha_k} \ket{D_k} \,
\end{equation}
as in Eq.~\eqref{eq:PPI-Dicke}, we then have
\begin{align}
    & \bra{\psi_{\mathrm{PPI}}^{(n)}}(\underbrace{Z\cdots Z}_{m}\underbrace{\mathds{1}\cdots\mathds{1}}_{n-m})\ket{\psi_{\mathrm{PPI}}^{(n)}} = \sum_{k=0}^n a_k^2 \, K^n_k(m), \notag \\[1mm]
    & \bra{\psi_{\mathrm{PPI}}^{(n)}}(\overbrace{Z\cdots Z}^{m}\overbrace{\mathds{1}\cdots\mathds{1}}^{n-m})(\overbrace{X\cdots X}^n)\ket{\psi_{\mathrm{PPI}}^{(n)}} \notag \\[-2mm]
    & \hspace{20mm} = \sum_{k=0}^n a_k\,a_{n-k} \,e^{i(\alpha_{n-k}-\alpha_k)} \, K^n_k(m).
\end{align}

Notice that $K^n_k(m) = (-1)^m K^n_{n-k}(m)$. Consider from now on even values of $m$ (change $m$ to $2m$ above), so that $K^n_k(2m) = K^n_{n-k}(2m)$; we then get
\begin{align}
    & \bra{\psi_{\mathrm{PPI}}^{(n)}}(\underbrace{Z\cdots Z}_{2m}\underbrace{\mathds{1}\cdots\mathds{1}}_{n-2m})\ket{\psi_{\mathrm{PPI}}^{(n)}} \notag \\
    & = \left\{ \begin{array}{ll}
        \sum_{k=0}^{\frac{n-1}{2}} q_k \, K^n_k(2m) & \ \text{for odd $n$} \\[1mm]
        \sum_{k=0}^{\frac{n}{2}-1} q_k \, K^n_k(2m) + a_{\frac{n}{2}}^2 \, K^n_{\frac{n}{2}}(2m)  & \ \text{for even $n$}
    \end{array} \right.
\end{align}
with $q_k := a_k^2+a_{n-k}^2$.
For $\bra{\psi_{\mathrm{PPI}}^{(n)}}(Z\cdots Z\mathds{1}\cdots\mathds{1})$ $(X\cdots X)\ket{\psi_{\mathrm{PPI}}^{(n)}}$ we obtain similar expressions, just replacing each $q_k$ by $r_k := 2a_k\, a_{n-k} \, \cos(\alpha_{n-k}-\alpha_k)$.

Now, the orthogonality conditions for the basis states constructed from the fiducial state $\ket{\psi_{\mathrm{PPI}}^{(n)}}$ through the action of $G_{\mathrm{tetra}}^{(n)}$ require that 
\begin{align}
    & \bra{\psi_{\mathrm{PPI}}^{(n)}}(\underbrace{Z\cdots Z}_{2m}\underbrace{\mathds{1}\cdots\mathds{1}}_{n-2m})\ket{\psi_{\mathrm{PPI}}^{(n)}} = \delta_{m,0} \notag \\
    \text{and } & \bra{\psi_{\mathrm{PPI}}^{(n)}}(\underbrace{Z\cdots Z}_{2m}\underbrace{\mathds{1}\cdots\mathds{1}}_{n-2m})(\underbrace{X\cdots X}_{n})\ket{\psi_{\mathrm{PPI}}^{(n)}} = 0
    \label{eq:orth_cstr_PPI}
\end{align}
for all $m=0,\ldots,\lfloor\frac{n}{2}\rfloor$; and by the PPI property, it is sufficient to verify these. Let us look at the cases of odd and even $n$ separately:

\begin{itemize}
    \item For odd $n$: from the two sets of constraints in Eq.~\eqref{eq:orth_cstr_PPI}, we obtain two independent linear systems of $\frac{n+1}{2}$ equations (for $m=0,\ldots,\frac{n-1}{2}$) each, for the $\frac{n+1}{2}$ variables $q_k$ and the $\frac{n+1}{2}$ variables $r_k$ ($k=0,\ldots,\frac{n-1}{2}$) respectively.
    E.g., the first system is $\big\{\sum_{k=0}^{\frac{n-1}{2}} q_k \, K^n_k(2m) = \delta_{m,0}\big\}_{m=0}^{\frac{n-1}{2}}$. Note that each $K^n_k(2m)$ is a polynomial in $m$ of degree $k$; hence, defining $P(m):=\sum_{k=0}^{\frac{n-1}{2}} q_k \, K^n_k(2m)$, we have that $P\in\mathcal{P}_{\leq\frac{n-1}{2}}$, the vector space of real polynomials in $m$ of degree $\le \frac{n-1}{2}$. The $\frac{n+1}{2}$ constraints in the system of equations above evaluate $P(m)$ at $\frac{n+1}{2}$ distinct points, and hence identify $P$ uniquely in $\mathcal{P}_{\leq\frac{n-1}{2}}$. Moreover, $\{K_k^n(2m)\}_{k=0}^{\frac{n-1}{2}}$ defines a basis of $\mathcal{P}_{\leq\frac{n-1}{2}}$ (because the degrees are $0,1,\ldots,\frac{n-1}{2}$, so the family is linearly independent), so that the representation of $P$ in this basis is unique. Therefore the coefficients $\{q_k\}$ are uniquely determined. The same reasoning applies to the second system of equations, from which we conclude that the coefficients $\{r_k\}$ are also uniquely determined.
    
    As it can be verified, these unique solutions are found to be $(\forall\,k)$
    \begin{align}
        q_k & = a_k^2+a_{n-k}^2 = \frac{1}{2^{n-1}}, \notag \\
        r_k & = 2a_k\, a_{n-k} \, \cos(\alpha_{n-k}-\alpha_k) = 0 \,.
    \end{align}
    Given the solution for $q_k = a_k^2+a_{n-k}^2$, one can introduce some real parameter $\theta_k$, for each $k=0,\ldots,\frac{n-1}{2}$, and write $a_k = \frac{1}{2^{(n-1)/2}}\cos\theta_k, a_{n-k} = \frac{1}{2^{(n-1)/2}}\sin\theta_k$.

    From the null value of $r_k$ above, we can see that if $a_k \,a_{n-k} \neq 0$, then one must have $\alpha_{n-k}-\alpha_k = \pm \frac{\pi}{2}$; up to changing for instance $\theta_k \leftrightarrow \pi-\theta_k$ (so as to change the relative sign of $a_k$ and $a_{n-k}$), we can always choose $\alpha_{n-k}-\alpha_k = + \frac{\pi}{2}$. If $a_k \,a_{n-k} = 0$ on the other hand, then one of the two phases $\alpha_k$ or $\alpha_{n-k}$ is irrelevant, and we can take the difference $\alpha_{n-k}-\alpha_k$ to have any value of our choice, for instance (again) $\alpha_{n-k}-\alpha_k = + \frac{\pi}{2}$. 

    All in all, without loss of generality the parameters $a_k$ and $\alpha_k$ of $\ket{\psi_{\mathrm{PPI}}^{(n)}}$ in Eq.~\eqref{eq:PPI-Dicke} can therefore be taken as in Eq.~\eqref{eq:psiPPI_amplitudes_phases}.
    
    \item For even $n$: similarly to the previous case, from the two sets of constraints in Eq.~\eqref{eq:orth_cstr_PPI}, we obtain two linear systems of $\frac{n}{2}+1$ equations (for $m=0,\ldots,\frac{n}{2}$) each, for the $\frac{n}{2}+1$ variables $q_k$ and $a_{\frac{n}{2}}$ for the first system, and the $\frac{n}{2}+1$ variables $r_k$ and $a_{\frac{n}{2}}$ for the second system ($k=0,\ldots,\frac{n}{2}-1$). 
    These systems are not independent anymore, as they share the variable $a_{\frac{n}{2}}$. Solving them independently gives the (unique, as above) solutions $a_{\frac{n}{2}}^2 = 1/2^{n-2}$ (and all other $q_k = 1/2^{n-1}$) for the first system, and $a_{\frac{n}{2}}^2 = 0$ (and all other $r_k = 0$) for the second system.
    Hence, we get a contradiction: there is no solution---i.e., no PPI state---that satisfies the orthogonality requirements in the case of even $n$.
    
\end{itemize}

\color{black}

\section{Clifford hierarchy and Cui-Gottesman-Krishna characterization of diagonal gates} \label{app:clifford}

\subsection*{Clifford hierarchy}

The Clifford hierarchy~\cite{Gottesman1999} is a recursively defined hierarchy of unitary operations that plays an important role in fault-tolerant quantum computation. Each level \( \mathcal{C}_k \) contains unitaries that conjugate Pauli operators into the previous level. For simplicity we will restrict the presentation here to the qubit ($d=2$) case.

To define the hierarchy, we begin with the $n$-qubit Pauli group \( \mathcal{P}_n \), which consists of all \( n \)-fold tensor products of the single-qubit Pauli matrices \( \{ \mathds{1}, X, Y, Z \} \), including phase factors \( \{ \pm1, \pm i \} \):
\begin{equation}
\mathcal{P}_n := \left\{ e^{i \phi} P_1 \otimes \cdots \otimes P_n \,\middle|\, P_j \in \{ \mathds{1}, X, Y, Z \},\; \phi \in \tfrac{\pi}{2} \mathbb{Z} \right\}.
\end{equation}
This group is closed under multiplication and forms the basis of the Clifford hierarchy.

The Clifford hierarchy is then defined recursively:
\begin{itemize}
    \item \( \mathcal{C}_1 := \mathcal{P}_n \), the Pauli group.
    \item \( \mathcal{C}_k := \left\{ U \in \mathcal{U}(2^n) \,\middle|\, U P U^\dagger \in \mathcal{C}_{k-1} \;\; \forall\, P \in \mathcal{P}_n \right\}. \)
\end{itemize}

For example, \( \mathcal{C}_2 \) is the Clifford group—unitaries that normalize the Pauli group—and \( \mathcal{C}_3 \) includes gates such as the \( T \)-gate, Toffoli, and other non-Clifford operations that are essential for universal quantum computation.

\subsection*{Diagonal unitaries in Cui-Gottesman-Krishna's framework}

In general, it is very difficult to characterize explicitly the unitaries in the Clifford hierarchy~\cite{silva2025}.
Following Ref.~\cite{Gottesman2017}, we will restrict ourselves to the classification of phase gates whose phases are (discrete) dyadic fractions of $\pi$---i.e., diagonal unitaries \( D_{f_m} \) acting as 
\begin{equation}
D_{f_m} \ket{\vec{z}} = \exp\left(i \frac{2\pi}{2^m} f_m(\vec{z})\right) \ket{\vec{z}},
\end{equation}
where \( f_m: \mathbb{Z}_2^n \to \mathbb{Z}_{2^m} \) is a polynomial over Boolean variables \( z_i \in \{0,1\} \), and \( m \in \mathbb{N} \) determines the phase precision.
The general form of \( f_m \) is:
\begin{equation}
f_m(z_1,\dots,z_n) = \sum_{S \subseteq [n]} a_S \prod_{i \in S} z_i \mod 2^m,
\end{equation}
where \( a_S \in \mathbb{Z}_{2^m} \) are integer coefficients. Since the global phase introduced by $D_{f_m}$ is irrelevant, one may choose to fix $f_m(0,\dots,0) = a_\emptyset = 0$, and therefore only consider nonempty subsets $S$ of $[n]$ in the sum above.

We now summarize Cui-Gottesman-Krishna's criterion~\cite{Gottesman2017} for determining the level \( \mathcal{C}_k \) to which \( D_{f_m} \) belongs.

\bigskip

\subsection*{Cui-Gottesman-Krishna characterization}

The Clifford level of a diagonal gate \( D_{f_m} \) is determined by the precision \( m \) and the degree of the polynomial \( f_m \). Ref.~\cite{Gottesman2017} shows that the level of the Clifford hierarchy in which a diagonal gate appears is determined by the largest value of \( m \) that appears in the exponent and the degree of the polynomial, and one only needs to consider the finest phase rotations. The level $k$ is then given by
\begin{equation} \label{eq:coarse}
k = \max_{S: a_S \neq 0} \big[ (m - 1) + |S| \big].
\end{equation}
This assumes all nonzero coefficients \( a_S \in \mathbb{Z}_{2^m} \) are odd. When one of the coefficients is divisible by a power of two, the effective precision is reduced. This can be seen by considering each monomial separately and expressing the coefficient in its 2-adic form.

Define the 2-adic valuation \( \nu_2(a) \) as the exponent of the highest power of 2 dividing \( a \in \mathbb{Z} \). For example:
\[
\nu_2(1) = 0,\quad \nu_2(4) = 2,\quad \nu_2(12) = 2,\quad \nu_2(0) := \infty.
\]
Each monomial of the form \( a_S \prod_{i \in S} z_i \) contributes to the diagonal unitary with a phase rotation of the form
\begin{equation}
\exp\left( i\frac{2\pi}{2^m} a_S \right) = \exp\left( i\frac{2\pi}{2^{m - \nu_2(a_S)}} b_S \right),
\end{equation}
where \( a_S = 2^{\nu_2(a_S)} b_S \) and \( b_S \) is odd. That is, the 2-adic valuation \( \nu_2(a_S) \) reduces the effective phase resolution for that term. This can be summarized in the following criterion:
\begin{equation} \label{eq:pacid}
k = \max_{S: a_S \neq 0} \big[ (m - \nu_2(a_S) - 1) + |S| \big].
\end{equation}

This ``maximum-over-terms'' rule reflects the fact that the Clifford hierarchy is closed under composition: multiplying two unitaries $U, V \in \mathcal{C}_k$ yields another unitary in $\mathcal{C}_k$. As a result, any diagonal gate whose phase polynomial contains multiple terms—each possibly contributing at a different level—will lie in the highest level among them. In other words, the term with the largest combined value $(m - \nu_2(a_S) - 1) + |S|$ dictates the overall hierarchy level of the gate, since one cannot ``undo'' the effect of a higher-level term by composing it with lower-level ones.

\end{document}